%% file: paper.tex
\pdfoutput=1

\documentclass{article}

\usepackage[utf8]{inputenc} 
\usepackage[T1]{fontenc}    
\usepackage{hyperref}       
\usepackage{url}            
\usepackage{booktabs}       
\usepackage{amsfonts}       
\usepackage{nicefrac}       
\usepackage{microtype}      
\usepackage{xcolor}         
\usepackage{paper}
\usepackage{adjustbox}
\usepackage{wrapfig}

\usepackage{makecell}

\newtoggle{arxiv}
\toggletrue{arxiv}

\iftoggle{arxiv}
{
\usepackage{natbib}
\setlength{\textwidth}{6.5in}
\setlength{\textheight}{9in}
\setlength{\oddsidemargin}{0in}
\setlength{\evensidemargin}{0in}
\setlength{\topmargin}{-0.5in}
\newlength{\defbaselineskip}
\setlength{\defbaselineskip}{\baselineskip}
\setlength{\marginparwidth}{0.8in}
}
{
\usepackage[nonatbib]{neurips_2020}
}
\usepackage{courier}
\usepackage{listings}
\definecolor{mygreen}{rgb}{0,0.6,0}
\definecolor{mygray}{rgb}{0.5,0.5,0.5}
\definecolor{mymauve}{rgb}{0.58,0,0.82}

\title{\bf \textsf{Holistic Cube Analysis: A Query Framework for \\Data Insights}}

\newcommand{\abstractcube}{\text{\rm AbstractCube}}
\newcommand{\btcube}{\text{\rm BaseTableGroupByCube}}
\DeclareMathOperator*{\cellsetaggr}{{\sf AbstractAggregate}}
\DeclareMathOperator*{\NaiveCubeCrawl}{{\sf NaiveCubeCrawl}}
\DeclareMathOperator*{\TopDownCubeCrawl}{{\sf TopDownCubeCrawl}}

\author{%
  {\sf \small Xi Wu, Shaleen Deep$^*$, Joe Benassi, Fengan Li$^\ddagger$, Yaqi Zhang$^\ddagger$, Uyeong Jang, James Foster, Stella Kim,} \\
  {\sf \small Yujing Sun, Long Nguyen, Stratis Viglas, Somesh Jha$^\dagger$, John Cieslewicz, Jeffrey F. Naughton$^\ddagger$}\\
  {\sf \small Google, $^*$Microsoft GSL, $^\dagger$UW-Madison, $^\ddagger$Celonis}
}
\date{}

\begin{document}

\maketitle

\input{abstract}

\section{Introduction}
\label{sec:intro}
\input{introduction}

\section{Holistic Cube Analysis framework}
\label{sec:framework}
\input{framework}

\section{Programming insights with HoCA}
\label{sec:power}
\input{power}

\section{Algorithmic space for HoCA}
\label{sec:considerations}
\input{considerations}

\section{Implementation and evaluation}
\label{sec:impl-eval}
\input{impl-eval}

\section{Applications and discussions}
\label{sec:discussions}
\input{discussion}

\section{Conclusion}
\label{sec:conclusion}
\input{conclusion}

{\small
\bibliographystyle{iclr2021_conference}
\bibliography{paper}
}

\appendix
\onecolumn

\section{Region Aumann-Shapley attribution for density metric}
\label{sec:density-model-derivation}
\input{aumann-shapley/density-aumann-shapley}

\section{Frequent Itemset Mining via cube crawling}
\label{sec:fim-model-code}
\input{fim-model-code}

\section{Comparison with OLAP}
\label{sec:olapcomparison}
\input{olapcomparison}

\section{More examples}
\label{sec:more-examples}
\input{more-examples}

\end{document}

%% file: abstract.tex
\begin{abstract}
\label{sec:abstract}
Many data insight questions can be viewed as \emph{searching} in a large space
of tables and finding important ones, where the notion of importance is defined
in some adhoc user defined manner. This paper presents {\bf Ho}listic {\bf C}ube {\bf A}nalysis (\textbf{HoCA}),
a framework that \emph{augments} the capabilities of relational queries for such problems.
HoCA first augments the relational data model and introduces a new data type AbstractCube,
defined as a \emph{function} which maps a region-features pair to a relational table
(a region is a tuple which specifies values of a set of dimensions).
AbstractCube provides a \emph{logical form} of data,
and HoCA operators are cube-to-cube transformations.
We describe two basic but fundamental HoCA operators,
\emph{cube crawling} and \emph{cube join} (with many possible extensions).
Cube crawling explores a region space,
and outputs a cube that maps regions to signal vectors.
Cube join, in turn, is critical for \emph{composition},
allowing one to join information from different cubes for deeper analysis.
Cube crawling introduces two novel \emph{programming} features,
(programmable) \emph{Region Analysis Models} (RAMs) and \emph{Multi-Model Crawling}.
Crucially, RAM has a notion of \emph{population features},
which allows one to go beyond only analyzing local features at a region,
and program \emph{region-population} analysis that compares region and population features,
capturing a large class of importance notions.
HoCA has a rich algorithmic space, such as optimizing crawling and join performance,
and physical design of cubes.
We have implemented and deployed HoCA at Google.
Our early HoCA offering has attracted more than 30 teams
building applications with it, across a diverse spectrum of fields
including system monitoring, experimentation analysis, and business intelligence.
For many applications, HoCA empowers novel and powerful analyses,
such as instances of recurrent crawling, which are challenging to achieve otherwise.
\end{abstract}


%% file: introduction.tex
Many data insight questions can be viewed as \emph{searching} in a large space
of tables (sometimes called slices) and finding important ones,
where the notion of importance is defined in some adhoc user-defined manner. 
For example, Figure~\ref{fig:most-important-contributros-from-base-table} starts with
a top-level CostPerClick change, and asks what are \emph{important} contributors to this change at slice level;
Figure~\ref{fig:most-important-outliers-from-base-table}, instead,
asks what are \emph{important} outliers in a large space of slices.
For both analysis questions, one needs to consider a space of regions
(a \emph{region} is a tuple which specifies values for a set of dimensions,
for example, [Device=Pixel6] is a region with 1 dimension,
[Device=Pixel6, Browser=Chrome] is region with 2 dimensions),
and for each region one needs to consider a set of \emph{features}
([IsToday, Cost, Clicks] for CostPerClick contributors,
and [Date, Revenue] for interesting-outliers).
Regions and features induce a space of tables,
and instead of only analyzing features of each region \emph{locally}, 
these tables usually need to be analyzed altogether, or holistically, to uncover insights.

We note that, with classic declarative queries (e.g., SQL, OLAP),
which are largely designed to answer forward-style questions
(i.e., from region and features to a table, such as ``Extract Revenue timeseries data at Browser=Chrome''),
one needs to construct complex and esoteric pipelines to answer the search questions
(such as the ones shown above).
To this end, we note that an array of systems (e.g.,~\cite{AKSGXSASM21, ABDGMNRS18, WDM15, WM13})
have been built to simplify \emph{data explanation} tasks.
However, these systems do not cover questions such as important-outliers
in Figure~\ref{fig:most-important-outliers-from-base-table},
even though the workflow is similar to CostPerClick-explanation
(Figure~\ref{fig:most-important-contributros-from-base-table}).

\begin{figure*}[htb]
  \centering
    \begin{subfigure}{0.48\textwidth}
    \centering
    \includegraphics[width=\linewidth]{
      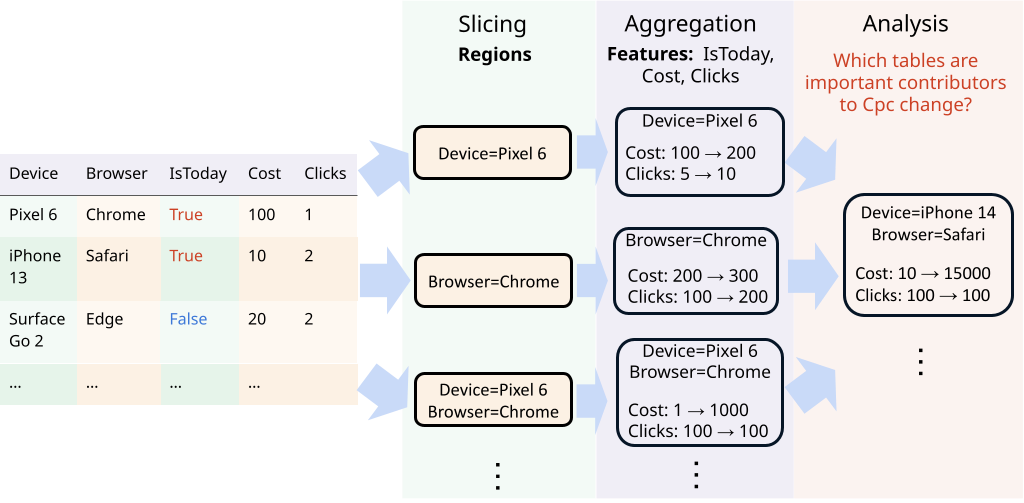}
    \caption{
      Finding important contributors to
      overall CostPerClick change, starting from a base table.
      While seemingly simple, a principled answer
      requires applying the AumannShapley method,
      see Section~\ref{sec:metric-attribution}.
    }
    \label{fig:most-important-contributros-from-base-table}
  \end{subfigure}\quad
  \begin{subfigure}{0.48\textwidth}
    \centering
    \includegraphics[width=\linewidth]{
    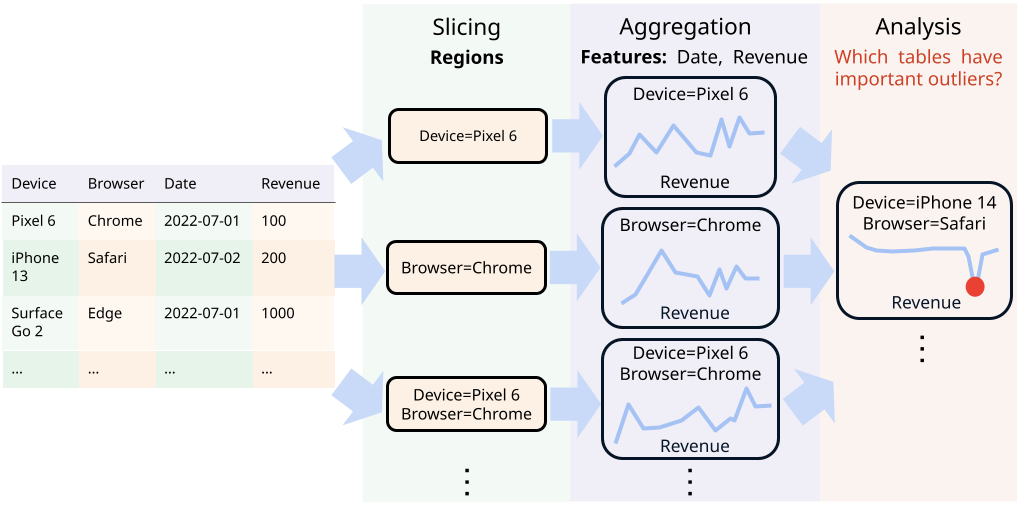}
    \caption{Finding important timeseries outliers,
      starting from a base table: Given a large space of slices,
      what are \emph{important} timeseries outliers one must pay attention to?
    }
    \label{fig:most-important-outliers-from-base-table}
  \end{subfigure}\quad\quad
  \caption{\small
    Many data insights problems can be viewed as:
    Given a \emph{space of tables} (often called slices)
    that are aggregated at different granularity,
    find interesting, or most attention-worthy slices
    (where these notions often need to be defined in an ad-hoc manner).
    Currently, to perform such analysis, analysts often need to build complex workflows to generate the tables,
    and then analyze them using advanced statistical and ranking techniques.
    By contrast, Figure~\ref{fig:hoca-programs-for-example-questions}
    shows that one can answer these questions using succinct HoCA programs.
  }
  \label{fig:example-hoca-questions}
\end{figure*}

\begin{figure*}[htb]
  \centering
  \begin{subfigure}{0.45\textwidth}
    \centering
    \includegraphics[width=\linewidth]{
    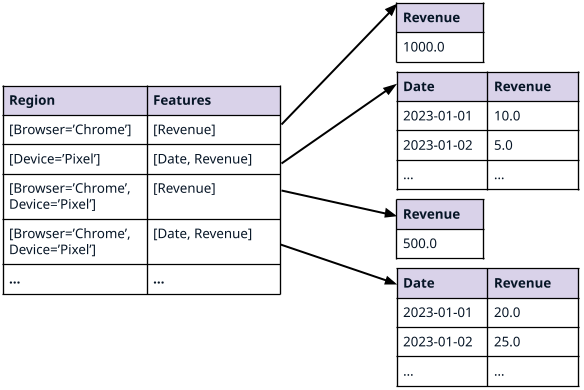}
    \caption{AbstractCube: A function mapping region-features to
    relational tables. It provides a logical form of the data. In particular,
    HoCA programmers do not need to care about how the function mapping is computed.
    There are different sub-types of AbstractCube, giving different ways to encode
    functions (see Section~\ref{sec:abstract-cube}). For the same logic cube,
    one can also have different physical representations.
    }
    \label{fig:hoca-cube-func}
  \end{subfigure}\quad\quad
  \begin{subfigure}{0.50\textwidth}
    \centering
    \includegraphics[width=\linewidth]{
      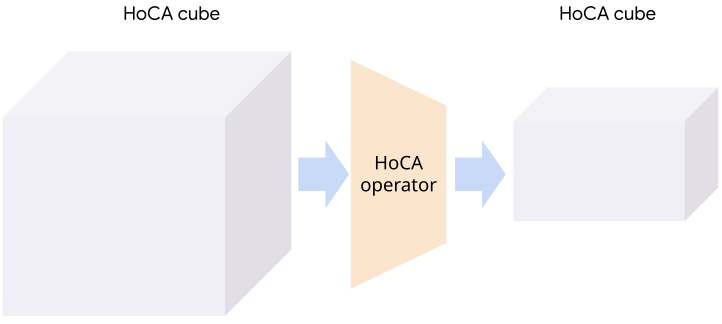}
    \caption{
      A HoCA operator: A cube-to-cube transformation
      (i.e., a higher-order, function-to-function transformation).
      A basic example of HoCA operator is \emph{cube crawling}:
      We traverse a region space, for each region do some analysis,
      and output a cube recording analysis results for regions.
      The ``function-as-data'' modeling of HoCA cubes is critical for \emph{composition}:
      Since the output is also a HoCA cube, we can apply further HoCA operators to the
      output cube (even though it may not be a projection-groupby-aggregation one).
      This gives rise to constructions such as Recurrent Cube Crawling
      (see Section~\ref{sec:recurrent-crawling}).
    }
    \label{fig:hoca-operators}
  \end{subfigure}
  \caption{\small
    \textbf{The HoCA paradigm}.
  }
  \label{fig:hoca-paradigm}
\end{figure*}

This paper presents \textbf{Ho}listic \textbf{C}ube \textbf{A}nalysis (\textbf{HoCA}),
a framework for systematically augmenting the capabilities of
relational queries for data insights.
Unlike the data-explanation systems mentioned above,
which basically augment SQL and still work on set-of-tuples,
HoCA first augments the relational data model by defining a new data type AbstractCube,
which is defined as a \emph{function} from RegionFeatures space to relational tables
(RegionFeatures space is a set of points, where each point specifies
a pair of region and features; Figure~\ref{fig:hoca-cube-func} gives a visualization).
AbstractCube provides a \emph{logical form} of data, and HoCA operators are cube-to-cube 
transformations, i.e., \emph{higher order}, function-to-function transformations
(Figure~\ref{fig:hoca-operators} provides a visualization).
Since input and output are of the same type, therefore, similar to relational algebra,
one can compose HoCA operators to form complex HoCA programs.

We describe two basic, but fundamental, HoCA operators:
\emph{Cube crawling} and \emph{cube join} (there are many more possibilities devising
more useful HoCA operators and extending the algebraic framework, which we discuss in
Section~\ref{sec:discussions}). Cube crawling introduces an enumeration-type operator,
which allows one to scan regions in a subspace, extract useful signals, and output
results into an output cube. Cube join, in turn, combines two input cubes into one,
by joining tables from two region-features spaces (see Figure~\ref{fig:cube-join}
for a visualization of the concept). Cube join is critical for \emph{composition},
since it allows one to join information from different cubes for deeper analysis,
and it is utilized in some of our most complex HoCA programs
(e.g., Recurrent Cube Crawling in Section~\ref{sec:recurrent-crawling}).

To solve a rich family of insight analytic problems,
we design cube crawling with two novel \emph{programming} features:
{\bf (1)} {\em Region Analysis Models (RAMs)},
which allows one to \emph{program} region analysis to be used in the crawl operator.
Crucially, RAM has a notion of \emph{population features},
which allows one to go beyond only analyzing local features at a region,
and program \emph{region-population} analysis that compares region and population features,
achieving a form of holistic analysis.
RAM abstraction allows easy integration and novel applications of advanced analytic techniques,
such as timeseries outlier detection (\cite{PyBSTS}), causal inference (\cite{BGKRS15}), and
differentiable programming (\cite{jax_2018_github}),
which goes much beyond previous work (e.g.~\cite{AKSGXSASM21}).
{\bf (2)} {\em  Multi-model crawling}, which allows one to apply multiple models,
potentially on {\em different} feature sets, during crawling.
Multi-model crawling significantly broadens the scope of analytic questions
one can answer with one crawl.

\begin{figure*}[htb]
  \centering
  \begin{subfigure}{0.45\textwidth}
    \centering
    \includegraphics[width=\linewidth]{
      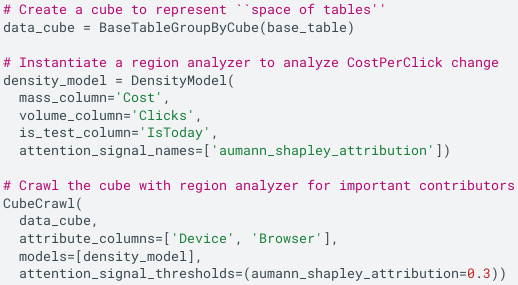}
    \caption{A HoCA program for finding important contributors to CostPerClick change.
      `aumann\_shapley\_attribution' is based on both region and population features.
      See Section~\ref{sec:metric-attribution} for details.}
    \label{fig:hoca-program-important-contributors}
  \end{subfigure}\quad\quad\quad
  \begin{subfigure}{0.45\textwidth}
    \centering
    \includegraphics[width=\linewidth]{
    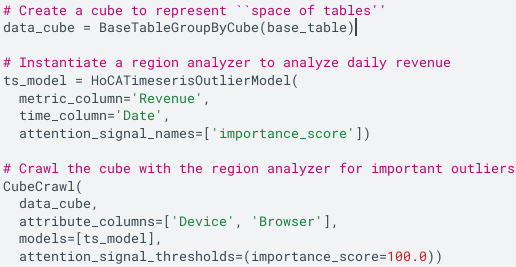}
    \caption{A HoCA program for finding important timeseries outliers.
    `importance\_score' is based on both region and population features. 
    }
    \label{fig:hoca-program-important-outliers}
  \end{subfigure}\quad\quad\quad
  \caption{\small
    HoCA programs that use cube crawling to answer the analytic questions raised in
    Figure~\ref{fig:example-hoca-questions}. In both cases, we analyze features from
    \emph{both the region and the entire population},
    which is holistic compared to analyzing region features alone.
  }
  \label{fig:hoca-programs-for-example-questions}
\end{figure*}

Together, these abstractions greatly simplified the analysis workflows aforementioned,
and also go further in capturing even more complex data insights problems.
Figure~\ref{fig:hoca-programs-for-example-questions} shows two HoCA programs that
answered formal versions of the two analytic questions in Figure~\ref{fig:example-hoca-questions}.
Each program only consists of three statements: Initializes a cube,
initializes an appropriate model, and then performs a cube crawling.
Since HoCA abstractions are compatible with standard relational constructions (e.g., SQL),
HoCA programmers can also construct more complex HoCA programs by composing HoCA operators,
together with standard relational constructions.

We have implemented and deployed HoCA at Google.
HoCA has a rich algorithmic space, such as optimizing crawling performance
(leveraging region space structure), optimizing cube join performance (local vs. global join),
and physical design of cubes (i.e., encoding of functions).
We briefly describe several implementations of cube crawling based on different foundations
(an in-house distributed relational query engine (\cite{F1Query}),
and also Apache Beam (\cite{PythonBeam})),
and evaluate their main performance characteristics.

HoCA has found many useful applications in system monitoring, experimentation analysis,
and business intelligence. For most of our applications, our users are system builders,
who use our new query framework to build data-insights systems.
Our early HoCA offering has attracted more than 30 teams directly dependent on our codebase.
In several applications, an interesting finding is that users had already built
adhoc solutions that can be captured by HoCA.
Migration to HoCA not only significantly simplified these constructions,
but also improved performance and scalability,
and enabled novel analyses (such as instances of recurrent crawling,
which would be challenging to achieve without HoCA concepts).
For Business Intelligence (BI) applications,
HoCA enabled fast development of novel data insights functionalities.
For example, somewhat surprisingly, we found that no previous work has even formulated
the problem of attributing a \emph{non-summable} metric change within a complex region space
(even for SUM/SUM ratio!). Within HoCA, we give a natural formulation
and develop a principled solution (Section~\ref{sec:metric-attribution}).

\textbf{HoCA vs. OLAP}.
Both HoCA and OLAP use the term ``cubes'', which indeed share some similarities.
For example, OLAP cubes can provide \emph{encoding} of HoCA cubes,
by providing mapping from region-features to relational tables.
A crucial difference, however, is that HoCA and OLAP examine
\emph{entirely different analytic questions}.
OLAP operators (e.g. drill-down, roll-up, etc.)
provide analytic capabilities in the ``forward'' direction
(e.g., region-features $\rightarrow$ table),
where users can leverage OLAP operators to perform an iterative, but manual,
slice-and-dice analysis to find data insights.
HoCA, on the other hand, provides operators to facilitate \emph{search} in the \emph{region space}.
For example in the CostPerClick attribution case described in Figure~\ref{fig:most-important-contributros-from-base-table},
HoCA starts with the \emph{top-level} CostPerClick change,
and the analysis deconstructs the metric and searches for the regions
that are major contributors to this change.
In this sense, HoCA operators are ``backwards''.
This difference in analytic questions induces vastly different considerations of
the cube design: In HoCA, cubes are \emph{logical form} of data,
and HoCA operators are higher-order, cube-to-cube transformations,
where they can be composed into a program,
with cubes flow in between \emph{explicitly} as data objects.
There is no counterpart of this in OLAP.
To this end, HoCA provides novel \emph{programming experiences} and \emph{constructions},
such as Region Analysis Model programming, cube join, and recurrent cube crawling.
A more detailed comparison is in Section~\ref{sec:olapcomparison} in~\cite{wu2023holistic}.

\textbf{Paper organization}.
Section~\ref{sec:framework} presents the HoCA framework of
HoCA cubes and operators.
Section~\ref{sec:power} gives various examples of programming
using HoCA to find data insights.
Section~\ref{sec:considerations} discusses important algorithmic considerations.
Section~\ref{sec:impl-eval} presents our current implementation of HoCA and its evaluations on several data sets.
Section~\ref{sec:discussions} discusses HoCA applications and future directions,
and we conclude in Section~\ref{sec:conclusion}.

%% file: framework.tex
This section presents the HoCA framework. We begin with some basic definitions:

$\bullet$ \textbf{Dimensions and measures}.
As in OLAP, we distinguish two kinds of columns: Dimensions and measures.
Dimensions are qualitative values that can be used to categorize and segment data.
Measures are numeric, quantitative values that one can measure.
For simplicity, in this paper, dimensions are synonymous with attributes,
and measures are synonymous with metrics.

$\bullet$ \textbf{Region}\label{def:region}.
A region is a tuple which specifies values of a set of dimensions
(e.g., (Device=`Pixel 5a') is a region,
(Device=`Pixel 5a', Browser=`Chrome') is another region).
For a region $\gamma$ and a dimension $d$,
we use $\gamma[d]$ to denote the value $d$ takes in $\gamma$.
We use $\dim(\gamma)$ to denote the set of dimensions in $\gamma$.
The number of dimensions in the region is called the \emph{degree} of the region
(e.g., (Device=`Pixel 5a') is degree-1,
(Device=`Pixel 5a', Browser=`Chrome') is degree-2).

$\bullet$ \textbf{Filtering a relation by a region}.
If $\bfB$ is a relation, we use $\bfB[\gamma]$ to denote
the set of tuples in $\gamma$ induced by a filtering:
For any $d$, $\bfB[\gamma]$ takes value $\gamma[d]$ for dimension $d$.
In SQL, this simply corresponds to a WHERE clause on $\bfB$.
For example, $\bfB[\text{Device=`Pixel 5a'}]$ considers tuples
where Device is `Pixel 5a'.

$\bullet$ \textbf{Features}.
A feature is a dimension or a measure.
For example, `Device' is a (dimension) feature,
and `Revenue' is a (measure) feature.
For example we can have $dims = [\text{`Device', `Browser'}]$, and
$measures = [\text{`Revenue'}]$.

With these, we discuss two main aspects of the HoCA framework:

\textbf{Abstract cube}.
In Section~\ref{sec:abstract-cube}, we define \abstractcube,
a data type defined as a \emph{function} mapping RegionFeatures space
to relational tables.
This function-as-data modeling allows us to simultaneously capture
a space of non-uniform tables in the co-domain of the function,
as well as the rich structure in the region space (e.g., hierarchies)
on the domain of the function. As a result,
operators on abstract cubes typically transform one cube to another.
That is, they are higher-order, function-to-function transformations.

\textbf{HoCA operators}.
In Section~\ref{sec:hoca-operators}, we describe two current HoCA operators:
\textbf{(i)} \emph{Cube crawling}, which explores a subspace of regions in a cube, 
extracting signals from each feature table visited,
and finally output a cube that maps regions to signal vectors.
\textbf{(ii)} \emph{Cube join}, which melds different cubes into one,
and is critical for composition.
\vskip -5pt
\subsection{AbstractCube: A function-as-data modeling}
\label{sec:abstract-cube}
\input{abstract-cube}

\subsection{HoCA operators}
\label{sec:hoca-operators}
\input{hoca-operators}

%% file: abstract-cube.tex
Abstract cubes are foundational data objects in HoCA.
Our definition was inspired by the traditional OLAP cubes,
but have diverged significantly.
Specifically, an abstract cube is a function from RegionFeatures space
to relational tables. We give details below.

\begin{definition}[\textbf{RegionFeatures space}]
\label{def:region-features-space}
The RegionFeatures space is a set of points,
where each point in the RegionFeatures space defines
a pair of region and features.
\end{definition}

For example, suppose that we have three dimensions, Device, Browser, Date,
and one metric, Revenue. The following are some examples of
RegionFeatures points:

{\small
(region: [Browser=’Chrome’], features: [Date, Revenue])\\
(region: [Device=’Pixel 7’], features: [Date, Revenue])\\
(region: [Device=’Pixel 7’, Browser=’Chrome’], features: [Revenue])
}

Basically, region gives the dimension values to zoom into,
while features gives the dimensions and metrics to analyze.

\begin{definition}[\textbf{AbstractCube}]
  An abstract cube is a function that maps the RegionFeatures space to
  relational tables. Specifically, given a (region, features) point,
  the cube maps this pair to a relational table whose schema conforms to
  features (i.e., columns of the table are dimensions and measures in features)
\end{definition}

What this definition is saying is that in HoCA, the \emph{logical form} of data is a function
of type ``region-features $\mapsto$ relational table''. There can be, however,
many different \emph{encoding} of such functions, depending on different needs. In the following
we give three examples:
\textbf{BaseTableGroupByCube} (a natural encoding using a base relation and SQL groupby-aggregation),
\textbf{CellsetCube} (an encoding leveraging cellset and $\cellsetaggr$, Algorithm~\ref{alg:cellset-aggr}),
and \textbf{RandomizedCube} (an example of non-base table aggregatable cubes)

\textbf{BaseTableGroupByCube}.
The simplest example of cubes is that we have a base table,
and given a point of region-features,
we compute a groupby-aggregation from the base view to the region and features. 
For example, consider the point (region: [Device='Pixel 5a'], features=[Date, Revenue]).
Then we filter by Pixel 5a, then GroupBy Date, and aggregate Revenue (so Browser is aggregated away).
As a result, we get relational table daily Revenue timeseries for Pixel 5a.
Algorithm~\ref{alg:basetable-aggr} gives more details.
\begin{algorithm}[!htb]
  \small
  \caption{\small GetView: Function that maps a region-features to a table
  for \btcube.}
  \begin{algorithmic}[1]
    \CLASSMEMBERS
      $\bfB$: A base relation;
      $\bfA$: A dictionary of `measure $\rightarrow$ aggregation function'
    \REQUIRE $\gamma$: A region; $F = (D, M)$: Dimension and measure features
    \ENSURE A view (relational table) $V$ \\
    \STATE Filter the base relation $\bfB$ using $\gamma$, and get
      $X = \bfB[\gamma]$.
    \STATE Using $X$, compute view $V$
      by GroupBy dimensions in $D$, and compute measure
      $m \in M$ using aggregation $\bfA[m]$.
    \RETURN $V$
  \end{algorithmic}
  \label{alg:basetable-aggr}
\end{algorithm}

\textbf{CellsetCube}.
CellsetCube is a more general encoding: First, it uses a generic value *
to augment dimension values so that we can encode more tuples,
which gives rise to cellset (Definition~\ref{def:cellset}).
Then, given a region (which does not allow * value) and features,
the mapping to a relational table is done by a \emph{lookup} in a cellset,
and associate appropriate cells into a table (again, does not contain * value).
We call this mapping $\cellsetaggr$ (Algorithm~\ref{alg:cellset-aggr}).

\begin{definition}[{\bf Cellset}]
  \label{def:cellset}
  Let $d_1, \dots, d_n$ be dimensions with domains $U_1, \dots, U_n$,
  the universe is $U=\bigcup_{i=1}^n U_i$.
  Any value in $U$ is called a \emph{specific} value.
  Let $*$ be a value not in $U$, which is called a \emph{generic} value,
  A cellset $\mathcal S$ is a function of the form
  ${\cal S}: \prod_{i=1}^n \big(U_i \cup \{*\}\big) \rightarrow \mathcal{M}$,
  where a \emph{cell} in a cellset a pair $(\kappa, {\cal S}(\kappa))$.
  Since $\cal S$ is a function, sometimes we also abuse the notation
  and call $\kappa$ a \emph{cell}.
\end{definition}

Algorithm~\ref{alg:cellset-aggr} defines the mapping from region-features to tables.
The idea is that use region to filter cells, then dimensions in the
features can take any \emph{specific} value, and all other dimensions take *.
\begin{algorithm}[htb]
  \small
  \caption{\small $\cellsetaggr$: Function that maps a region-features to a table
  for CellsetCube.}
  \begin{algorithmic}[1]
    \CLASSMEMBERS $\cal S$: A cellset
    \REQUIRE $\gamma$: A region;
      $F = (D, M)$: features  of dimensions and measures.
    \ENSURE A view (relational table) $V$
    \STATE Instantiate a cellset $K={\cal S}$
      \quad{\color{green!90} \# Initially, all cells}
    \STATE $K = K[\gamma]$. That is, for any $d$ that $\gamma[d]=v$,
      we only keep cells that take value $v$ at dimension $d$
      \quad{\color{green!90} \# Filtering using the region}
    \STATE Retain a cell in $x \in K$ using the following rules:\quad{\color{green!90} $(D, M)$ features}
      \begin{enumerate}
        \item[i.] $\forall d \in D,\;x.d \neq *$
          \quad{\color{green} \# A dimension feature takes any \emph{specific} value}
        \item[ii.] $\forall d \notin D \cup \dim(\gamma),\; x.d = *$
          \quad{\color{green!90} \# Any other dimension takes *}
          \quad{\color{green!90}}
      \end{enumerate}
    \STATE Remove any dimension in the $K$ that has $*$ value,
      the result is a usual relation with no generic value $*$.
      We call this relation $V$.
    \RETURN $V$.
  \end{algorithmic}
  \label{alg:cellset-aggr}
\end{algorithm}

\textbf{RandomizedCube: A non-base table aggregatable cube}.
We note that there are cubes that are \emph{not} base-table aggregatable.
This is because one can construct a cellset with a special measure,
where there is no relationship among the cells w.r.t the measure.
One example is RandomizedCube: We start with a set of cells,
and we map each cell to a randomized value (we call this value `randomized\_metric'),
which is sampled independently for every cell.
In other words, For every cell, say (Device=*, Browser=*, Date=*),
the measure {\rm `randomized\_metric'} is a value sampled independently from a random value generator.
Clearly, this gives a valid cellset, and so we can create a CellsetCube from it. 
However, this cube is not base table aggregatable,
because the randomized value of a cell is {\em not}
determined by cells that take specific values in dimensions.
  
This example may seem artificial, but it shows that there can be data
cube objects that are information-theoretically non-base-table aggregtable.
In practice, HoCA such as cube crawling may generate results
that is naturally a cube, but the cost of generating the cube is
so high that they should be treated similar to the random cube:
We should simply record these results in some cube data object for further analysis
(An example of this is the intermediate cubes in the recurrent cube crawling pipeline,
see Section~\ref{sec:recurrent-crawling}). Therefore, these cubes are not instances of \btcube,
but they are instances of CellsetCube.

%% file: hoca-operators.tex
With abstract cubes, we can then design operators to analyze them.
We present two basic HoCA operators: Cube crawling and cube join.

\textbf{Cube crawling}.
Cube crawling refers to the following procedure in general:
Given a \emph{space of regions} (region space),
and a set of \emph{Region Analysis Models}
(\textbf{RAMs}, or simply models),
we want to traverse every region in a region space,
and for each region we extract features needed by the models, 
and evaluate models over these features to get a signal vector.
Finally, we output a cube mapping regions to signal vectors.
Algorithm~\ref{alg:theoretical-crawling} gives the conceptual algorithm.
We note that:

\begin{algorithm}[htb]
  \small
  \caption{\small $\NaiveCubeCrawl$: A na\"{i}ve algorithm for cube crawling, for illustrating the concept.
    Note that models are programmable, and are fed as input to crawling. This algorithm is impractical
    since the region space may be exponentially large.}
  \begin{algorithmic}[1]
    \REQUIRE $\bfC$: $\sf \color{blue} AbstractCube$;
      {\bf region\_space}: A specification of region space;
      {\bf models}: $\sf \color{blue} List[Model]$
    \ENSURE  $\bfR: {\sf \color{blue} AbstractCube}$
      A cube mapping regions to signals. For simplicity,
      here the output cube is encoded as a set of
      region-signals pairs.
    \STATE {\color{green!90} \# Set population features for each model}
    \FOR{$f$ \textbf{in} models}
      \STATE $f$.population\_df = \bfC($f$.attribute\_features, $f$.metric\_features)
    \ENDFOR
    \STATE $\bfR = \emptyset$.
    \STATE {\color{green!90} \# Evaluate models on every region and collect results.}
    \FOR{region \textbf{in} region\_space}
      \STATE subcube = $\bfC$[region]
      \STATE $r = \emptyset$
      \FOR{$f$ \textbf{in} models}
        \STATE data\_frame = subcube($f$.attribute\_features, $f$.metric\_features)
        \STATE $r$.append($f$(data\_frame))
      \ENDFOR
      \IF{$r$ passes all filtering}
        \STATE $\bfR$.append((region, $r$))
      \ENDIF
    \ENDFOR
    \RETURN $\bfR$.
  \end{algorithmic}
  \label{alg:theoretical-crawling}
\end{algorithm}

\textbf{\em Region space}. Region space specification specifies the space of regions for crawling.
  The simplest example of region space is the \emph{cube space}, specified as `cube\_space(dimension\_list)'
  such as \\`cube\_space([Device, Browser])', which considers \emph{all} possible combinations of the two dimensions,
  and generate regions such as (Device=`Pixel 5a'), (Device=`Pixel 5a', Browser=`Chrome'), \\(Browser=`Chrome').
  Note that there are various possibilities to shrink the region space by leveraging the data topology.
  For example we can consider data topology such as product of hierarchies (\cite{RSY18, FGKNST05}),
  and functional dependencies (\cite{AKSGXSASM21}), which can be easily encoded as region space specification.

\textbf{\em Model features}.
Each model $f$ has a set of dimension and measure features,
encoded as attribute\_features and measure\_features.

\textbf{\em Population feature data for a model}.
  Lines 2--4  sets the population data for each model,
  by extracting features from $\bfC$.

\textbf{\em Region feature data for a model}.
  At lines 10--12, we zoom into a region, which gives a subcube,
  and we extract model features from the subcube,
  which gives region features. Therefore, each model $f$ sees two pieces of data, region features and population features.

\textbf{\em Model programmability}.
  Cube crawling supports \emph{programming new region analysis models}
  and use them for crawling. This feature allows us to apply different models and
  capture a wide variety of different use cases.
  In Section~\ref{sec:power}, we describe in details how model programming is designed,
  and give detailed examples.

\textbf{\em Crawling output}.
The output of crawling is a cube with a particular structure:
It maps a region (no feature) to a signal vector
(signals are encoded as a table with one row,
and each column encodes one signal name).
One can naturally extend this to mapping regions to relational tables in general.

\begin{example}[\textbf{Crawling for timeseries outliers}]
  \label{example:crawling-ts-outliers}
  This is a first example of cube crawling. We first instantiate a model:

  {\rm
  \begin{lstlisting}[language=Python]
  ts_model = TimeSeriesOutlierDetector(
    datetime='Date', metric='Revenue',
    timepoint='2022-01-31',
    signals=['importance_score'])\end{lstlisting}
  }
  This model then instantiates where attribute\_features
  is [Date], and metric\_features is [Revenue].
  Finally, given a timeseries, we want to check whether the metric
  value at 2022-01-31 is an outlier, this checking is captured by
  computing a signal called `importance\_score', which is offered by
  the model programmer. Then we can crawl with this model on the cube \bfC:

  {\rm
  \begin{lstlisting}[language=Python]
  CubeCrawl(
    C,
    region_space=span(Device, Browser), 
    models=[ts_model])\end{lstlisting}
  }
  Basically, this crawling considers regions spanned by Device and Browser
  (e.g., a single device where browsers are aggregated, a single browser where
  devices are aggregated, or a combination of a specific device and a specific
  browser). For every region, we extract a timeseries of Date and
  Revenue, then the model \textbf{ts\_model} analyzes this timeseries,
  and outputs a signal `importance\_score'.
\end{example}

\begin{figure}[htb]
  \centering
  \includegraphics[width=0.75\textwidth]{
    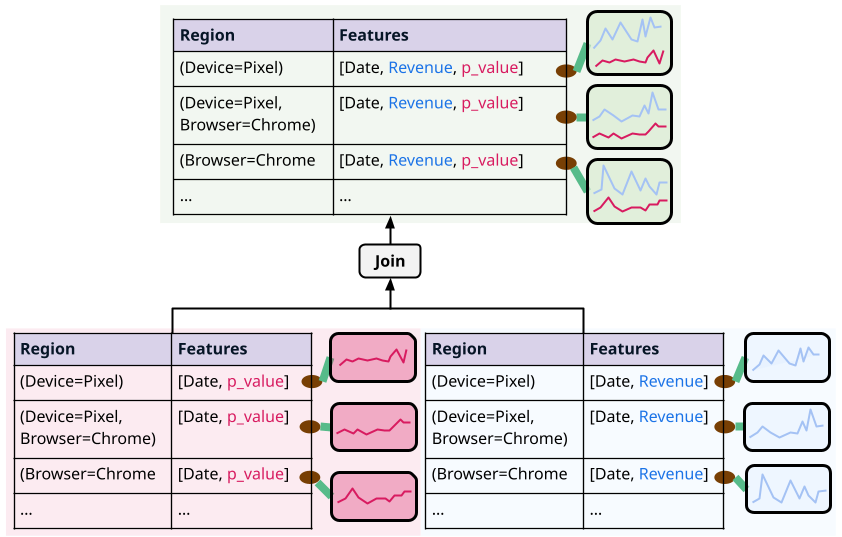}
  \caption{\small
    \textbf{Cube join}: The left (red) and right (blue) 
    are two cubes (RegionFeatures$\rightarrow$RelationalTable functions).
    The cube join operator allows one to meld these two functions into one,
    by matching region, and then join on dimension features.
    The joined cube (green) thus contains both p\_value and Revenue metric features.
  }
  \label{fig:cube-join}
\end{figure}

\textbf{Cube join}.
Cube join operator joins two cubes into one. Cube join is critical for \emph{composition},
since it allows one to join information from different cubes for deeper analysis.
For simplicity of presentation,
we only consider the case where two cubes to join have the same dimensions.
There are two ways to define cube join: A local definition and
a global definition (if we choose to work with Cellset cubes as the abstract cubes).
The local definition is based on the observation that the core of a cube
is to implement mapping given region and features.
This gives the following definition:

\begin{definition}[\textbf{Local definition of cube join}]
  Let $\bfC_l$ and $\bfC_r$ be two cubes.
  We define $\bfC = \bfC_l\ \sf{JOIN}\ \bfC_r$, the \emph{joined cube},
  where the $\sf GetView$ is implemented as follows.
  Let $\gamma$ be a region, and $\mathcal{F}$ be features, then
  \begin{enumerate}
    \item We extract $\text{df}_l = \bfC_l[\gamma](\mathcal{F})$
    \item We extract $\text{df}_r = \bfC_r[\gamma](\mathcal{F})$
    \item We compute $\text{df} = \text{df}_l\ {\sf JOIN}\ \text{df}_r$,
    where the join keys are the dimension feature columns.
  \end{enumerate}
  We can use different join semantics for the $\sf JOIN$ operator
  ($\sf LEFT\_JOIN$, $\sf RIGHT\_JOIN$, $\sf INNER\_JOIN$).
  Figure~\ref{fig:cube-join} visualizes this process.
\end{definition}

On the other hand, if we are working with CellsetCube,
we can define cube join at the cellset level,
and then use the abstract aggregation (Algorithm~\ref{alg:cellset-aggr})
to induce an cube instance. This gives the following:

\begin{definition}[\textbf{Global join for cellsets}]
  Let $\bfC_l$ and $\bfC_r$ be two cubes. Suppose that both cubes have been
  converted into instances of CellsetDataCube. Then we define
  $\bfC = \bfC_l\ \sf{JOIN}\ \bfC_r$ the \emph{joined cube}, as follows:
  (1) We join the underlying cellsets of $\bfC_l$ and $\bfC_r$,
    using the specified join operator
    ($\sf LEFT\_JOIN$, $\sf RIGHT\_JOIN$, $\sf INNER\_JOIN$),
    where the join keys are all the dimensions.
    Let the resulting cellset be ${\cal S}_{\text{joined}}$.
  (2) Return $\bfC = \text{\rm CellsetDataCube}({\cal S}_{\text{joined}})$.
\end{definition}

\begin{example}[\textbf{Augmenting a cube with historical {\bf p\_value}}]
\label{example:augmenting-cube}
We give more explanation to Figure~\ref{fig:cube-join}:
Suppose that we run outlier detection crawling daily,
and the outlier detector computes a {\bf p\_value} for each region.
These computed p-values (for each day) is stored in a cube $\bfC'$
(i.e., red cube on the left in Figure~\ref{fig:cube-join}),
which has dimensions {\bf Device}, {\bf Browser}, {\bf Date}, and measure {\bf p\_value}.
Now, on a new date $t^*$, and we have a cube (blue cube on the right in Figure~\ref{fig:cube-join})
$\bfC$ with dimensions {\bf Device}, {\bf Browser}, {\bf Date}, and measure {\bf Revenue},
then we can ``inner join'', $\bfC \text{ INNER\_JOIN } \bfC'$,
using join key {\bf Device}, {\bf Browser}, and {\bf Date}.
This gives a larger cube now with two measures {\bf Revenue} and {\bf p\_value}.
In Section~\ref{sec:recurrent-crawling}, we discuss recurrent cube crawling,
where a revised timeseries outlier detector can leverage the historical
`p\_value' in the joined data cube.
\end{example}

This construction of augmenting a cube with previous analysis results
is useful in continual monitoring, and in particular in
Recurrent Cube Crawling, described in detail in Section~\ref{sec:recurrent-crawling}.

%% file: power.tex
This section gives various examples of programming with HoCA to uncover data insights.
The power of the current HoCA framework comes from several aspects:
{\bf (1)} As we mentioned before, one can program region analysis models (RAMs)
to collect relevant analysis on a set of data features into a module and reuse.
Crucially, RAM has a notion of \emph{population features},
which allows one to go beyond only analyzing local features at a region,
and program \emph{region-population} analysis that compares region and population features,
achieving a form of holistic analysis.
RAM abstraction allows easy integration and novel applications of advanced analytic techniques,
such as timeseries outlier detection (\cite{PyBSTS}), causal inference (\cite{BGKRS15}),
differentiable programming (\cite{jax_2018_github}),
which goes much beyond previous works (e.g.~\cite{AKSGXSASM21}).
{\bf (2)} An important feature of our crawler API design is that it allows crawling with
\emph{multiple models}, potentially requesting different feature sets, during one crawling.
This significantly broadens the scope of analytic questions we can answer with one crawl.
{\bf (3)} Since input and output are both cubes,
HoCA operators can be composed  to construct complex HoCA programs for advanced use cases.

\subsection{Programming RAMs}

RAM abstraction allows easy integration and novel applications of advanced analytic techniques,
such as timeseries outlier detection (\cite{PyBSTS}), causal inference (\cite{BGKRS15}),
differentiable programming (\cite{jax_2018_github}). We give a few examples.

\subsubsection{Metric attribution via Aumann-Shapley}
\label{sec:metric-attribution}

Let $\bfC$ be a cube and $M$ be a metric.
For a region $\gamma$, we use $M^\gamma$ or $M[\gamma]$ to denote the
value of metric at the region. Specifically, for empty region $\gamma=[]$,
we use $M^p$ to denote the population metric value.
Suppose now that we divide the data into two segments (test segment,
and control segment), where on the test segment we have metric value $M^p_t$,
and on the control segment we have metric value $M^p_c$, so there is a
\emph{population metric value change} $\Delta^p = M^p_t - M^p_c$.
In the \emph{metric (change) attribution} problem, we want to attribute
$\Delta^p$ to regions in the cube: That is we want to assign a score,
denoted as $\ras(\gamma)$, and naturally we want this scoring to have the
following natural properties: \textbf{(1) Completeness}: That is for the
empty region the attribution score is exactly $\Delta^p$, $\ras([]) = \Delta^p$.
\textbf{(2) Additivity:} For any two disjoint regions $\beta, \gamma$,
$\ras(\beta \cup \gamma) = \ras(\beta) + \ras(\gamma)$. So for example,
we have many different devices and their attributions should add up to the
population level change. \textbf{(3)} $\ras(\gamma)$ should reflect the \emph{importance}
of contribution to $\Delta^p$. Note that (3) is informal, and it is reflected in our algorithms.
We found that a novel application of the Aumann-Shapley method (\cite{SS11,AS2015})
gives principled answers to such questions. Our approach uses the following definition

\begin{definition}[\textbf{Region-Ambient Metric Model}]
A region-ambient metric model combines metric values (maybe more than that of $M$) of a region $\gamma$,
and metric values of the ambient of the region $\gamma^c$
(i.e., data not in $\gamma$), to recover the population metric.
\end{definition}

If $M$ is summable (e.g., Revenue), a region-ambient metric model is simple:
$F(w, w') = w + w'$ ($w$ models the metric value in a region, and $w'$ is the metric value
in the ambient of the region). If $M$ is non-summable, we may need more metrics than just $M$:
For example, for a density metric $M = W/S$ where both $W$ and $S$ are summable, we need
values of both $W$ and $S$ in order to recover the population metric value,
and a natural region-ambient metric model has four parameters:
$F(w, w', s, s') = (w+w')/(s+s')$.

With a region-ambient metric model, one can then apply the Aumann-Shapley method to the model,
with test/control metric values as two end points.
This can be easily illustrated on a summable metric ($F(w,w') = w+w'$):
In the control segment we have point $P_0 = (M^\gamma_c, M^p_c - M^\gamma_c)$,
and in the test segment we have point $P_1 = (M^\gamma_t, M^p_t - M^\gamma_t)$.
We consider a straight line $h(t) = (1-t)P_0 + tP_1, t \in [0,1]$.
Then for $G(t) = F(h(t))$, we have:
\begin{align}
\begin{split}
   \Delta^p = G(1) - G(0)
   = &\int_0^1 G'(t)dt \\
   = &\int_0^1 \left(\frac{\partial F}{\partial w}\frac{\partial w}{\partial t} + \frac{\partial F}{\partial w'}\frac{\partial w'}{\partial t} \right) dt \\
   = & \left(\int_0^1 \frac{\partial F}{\partial w}\frac{\partial w}{\partial t} dt\right)
      + \left(\int_0^1 \frac{\partial F}{\partial w'}\frac{\partial w'}{\partial t} dt\right)
\end{split}
\end{align}
where the last equality switch sum-integration order.
We then define $\ras(\gamma)$ as $\int_0^1 \frac{\partial F}{\partial w}\frac{\partial w}{\partial t} dt$,
which is easily evaluated to $\ras(\gamma) = M^\gamma_t - M^\gamma_c$,
which is just our intuitive notion of metric shift in the region!

For non-summable metrics, this method works the same way,
but produces formulas that are far less intuitive.
For example, for the density metric, if we use the metric model
$F(w, w', s, s') = (w+w')/(s+s')$, the attribution formula for a region,
which is the result of doing path integration, becomes:
\begin{align}
  \label{eq:density-region-attribution}
  \begin{split}
    &\ras(r) \\
    = & \frac{(w_t^r-w_c^r)(s_t^p-s_c^p) -
          (w_t^p-w_c^p)(s_t^r-s_c^r)}{(s_t^p-s_c^p)^2}
        \Big(\ln(s_t^p) - \ln(s_c^p)\Big)\\
      & + \frac{s_t^r-s_c^r}{s_t^p-s_c^p}
        \Big(\frac{w_t^p}{s_t^p} - \frac{w_c^p}{s_c^p}\Big)
  \end{split}
\end{align}
which is much less intuitive, unless one approaches from first principles.
Due to space, the derivation and the proofs that they satisfy
completeness and additivity are deferred to Section~\ref{sec:density-model-derivation}.
For more general metric models, where closed forms of integration are
more difficult to derive, one can use differentiable programming framework,
such as JAX (\cite{jax_2018_github})), to compute the integration numerically.

\subsubsection{Outlier detection and self supervision}
\label{sec:outlier-detection-and-ssl}
\textbf{Adapting advanced timeseries outlier detector for region analysis}.
because model programmers can directly work with features without caring
how they are generated, we can adapt various timeseries outlier analysis
(e.g.~\cite{BGKRS15} and its variants) almost verbatim as region analysis models.
The availability of population features also enables us to improve these analyses:
These timeseries analyses primarily focus on the \emph{shape} of one timeseries.
However, the cube population data allows us to measure the \emph{normalized size}
of a timeseries within the population. This allows us to compute signals that
incorporate both the normalized size and outlier score.
Due to the normalization, these hybrid signals are easier to filter
than unnormalized versions, and they also provide better signals for larger
slices that are experiencing anomalies, due to the hybrid signal design.

\textbf{Calibrating p-value using historical p-values}.\label{para:self-supervised-outlier-detector}
In Example~\ref{example:augmenting-cube} we described joining a cube
with historical `importance\_score'. With the joined cube,
we can then leverage the these historical importance scores for outlier detection,
which provide important information for improving precision/recall of the detection for the current time point.
Specifically, we revised the CausalImpactModel (\cite{BGKRS15}) for outlier
detection to make use of historical p-value:
Causal impact detection makes a statistical assumption that that p-value
follows $\text{Unif}(0,1)$ distribution under null hypothesis.
However, this assumption does not really hold on practical data.
Instead of working under this assumption, we then use the historical p-values
(enabled by cube join) to calibrate and get the final p-value.
Note that this can be viewed as a form of self-supervision, because we are using
the p-values learned to supervise model behavior.

\subsubsection{Emulating previous constructions}
\label{sec:emulation}
HoCA can also be used to emulate various classic constructions.
As a simple example, Section~\ref{sec:fim-model-code} in~\cite{wu2023holistic}
describes a region analysis model, FrequentItemsetModel,
to use in cube crawling to solve FIM.

\textbf{Emulating DIFF (Abuzaid et al. (2021))}.
We note that data explanation engines described in~\cite{AKSGXSASM21}
can be implemented as a single cube crawling with proper region analysis models.
Specifically, the authors in~\cite{AKSGXSASM21} defined $\sf DIFF$ operator,
which receives two data frames, and differentiate these two data,
based on some simple statistics, to find interesting regions.
One can encode this as a cube crawling by using a column `is\_test',
where $\sf True$/$\sf False$ identifies the two tables to differentiate.
Moreover, computing the differentiation statistics can be wrapped into a region analysis model,
which loads `is\_test' column as a feature.
This way, we can \emph{emulate all use cases supported in } $\sf DIFF$.
We note that the DIFF paper described support for UDFs for computing differentiation statistics.
A key difference here is that the DIFF approach does not have the RAM abstraction,
and so there is no separation between features and analysis of features.
This causes difficulties in integrating with advanced statistical analyses,
and applying multiple analyses during one crawl,
since different analyses may need features aggregated at different granularity.

\subsection{Multi-model crawling}
\label{sec:multimodel}

Another important feature of our crawler interface design is that allows one to combine multiple models,
potentially requesting different feature sets, during one crawl.
Multi-model crawling significantly expands the scope of analytic questions
that can be answered using one crawl.
For example, in our application of
finding important timeseries outliers,
it is usually useful to combine two models, EntityWeightModel and TimeSeriesOutleirModel,
where the EntityWeightModel checks the weight of the timeseries
(for example this could be the total revenue in the time period),
and only if the region timeseries passes some thresholds
on the weight we invoke the timeseries outlier detection.
In this case, since EntityWeightModel is much cheaper than outlier detector,
it can save a lot of computation by skipping ``tiny'' slices.
Note that in this example,
EntityWeightModel only requests feature [`TotalRevenue'],
while the TimeSeriesOutlierModel requests features [`Date', `TotalRevenue'].
Without multi-model crawling, one would need to first do a crawl using
EntityWeightModel to analyze TotalRevenue, then using the results to filter
timeseries and then analyze timeseries for outlier, which results in much
more complex analysis programs.

We note that this design is also flexible in two ways:
{\bf (1)} By having EntityWeightModel requests [`Date', `TotalRevenue'] as features,
we can immediately have more fine-grained filtering (e.g., on maximum daily total revenue),
and {\bf (2)} EntityWeightModel can also be combined with other models
for different analysis.

\subsection{Compositions}
\label{sec:compositions}
\input{compositions}

%% file: compositions.tex
Our discussion so far only covers, basically, what a \emph{single}  cube crawling can do.
This section gives two examples of complex analysis pipelines formed by
composition of operators:
\textbf{(1)} An analysis pipeline that is a \emph{sequential composition} of
two cube crawling operators.
\textbf{(2)} A \emph{recurrent} analysis pipeline which
leverages all the constructions we discussed so far: At each timepoint a cube crawling is invoked,
it leverages \emph{previous crawling results} as features.
More examples can be found in Section~\ref{sec:more-examples} in~\cite{wu2023holistic}.

\subsubsection{Composition of two crawlings}
\label{sec:composing-two-crawlings}

In Example~\ref{example:crawling-ts-outliers},
we used a cube crawl to find regions that have important outliers.
Now, given the result cube, we may ask a \emph{further question}:
\emph{``What regions are \emph{frequently} involved in regions that have important outliers in this cube?''}
For example, we may find that `(Device=Chrome)', which is itself a region,
appears in many regions that have important outliers
(e.g., (Browser=Chrome), or (Device=Pixel 6a, Browser=Chrome), or (Device=Desktop, Browser=Chrome)).
Such an analysis may unveil that Chrome may need more attention,
because its combination with many devices have resulted in anomalous behavior.
We note that this question can be answered with a second crawl,
working on the result cube from the first crawl:
\textbf{(1)} Recall the output of the first crawl is a cube which only keeps
regions that have interesting outliers. We augment the output cube with a measure column
which measure the ``weight'' of a region in terms of outlier
(the simplest measure is to put a constant 1 for every region).
The augmented cube is the input of next crawl.
\textbf{(2)} We can then apply techniques such as Frequent Itemset Mining
(which can be implemented using one crawl to uncover patterns that frequently
(w.r.t. the weight measure) show up in regions that have important outliers.

\subsubsection{Recurrent cube crawling}
\label{sec:recurrent-crawling}

In the above, we discussed using cube join to inject previous crawling results
as features for later detection.
Figure~\ref{fig:recurrent-cube-crawling} demonstrates a recurrent
cube crawling pipeline applying this analysis in continual monitoring:
Suppose that at time $t^*$ we want to do a crawling to find interesting outliers.
Results of previous crawlings are recorded into a cube $\bfR_{t}$ (red)
(in particular, $\bfR_t$ \emph{appends} results at $t$),
so $\bfR_{t^*-1}$ contains all historical crawling results for $t \le t^*-1$.
Now at $t^*$, we have a data cube $\bfC_{t^*}$ of newly collected data
(e.g., Revenue) for detection, we join these two cubes to get
$\bfD_{t^*} = \text{Join}(\bfR_{t^*-1}, \bfC_{t^*})$ (green),
which is then fed to a cube crawling for analysis.
In our applications, we found such constructions to be useful for outlier detection quality.
We give two simple examples:
\textbf{(1) Self-supervision in the continual monitoring}.
We described above~\ref{para:self-supervised-outlier-detector} of a region analysis
models that leverages historical p-value for self supervision.
Recurrent crawling is a natural fit for applying such models,
where the cube join prepares historical p-values in the cube to crawl.
\textbf{(2) Removing transient anomalies}.
Advanced outlier detectors all learn the shape of timeseries
(e.g., Bayesian Structural Time Series) from existing data points.
However, sometimes timeseries may contain transient anomalies that only appears
short period of time and disappeared (e.g., several minutes).
Leaving these anomalies in the timeseries will bias the learning towards
that these values are ``normal'' (i.e., there is a new trend), and can cause noise.
With recurrent cube crawling, we can simply detect transient anomalies and record
them in the red cube, and upon a new detection, we use these labels to locate
transient anomalies and replace them with normal values. This allows us to
systematically correct transient anomalies.

\begin{figure}[htb]
  \centering
  \includegraphics[width=0.75\textwidth]{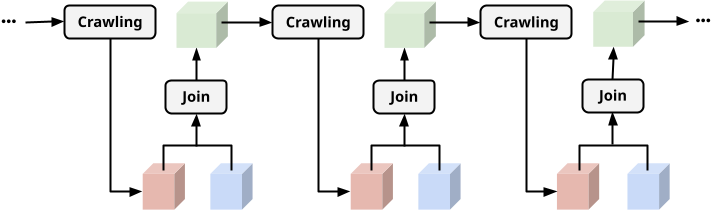}
  \caption{\small
  Recurrent cube crawling: Previous crawling results are joined with
  new data (e.g., revenue cube of today) for analysis.
  For outlier detection, it is very useful to use the p-values of
  previous detection results to calibrate the current detection.}
  \label{fig:recurrent-cube-crawling}
\end{figure}

%% file: considerations.tex
The last two sections covered the HoCA framework, and programming
data insights with HoCA abstractions.
This section turns to the algorithmic space of HoCA:
Section~\ref{sec:considerations:cube-crawling} considers algorithms for cube crawling,
Section~\ref{sec:considerations:join} on algorithms for cube join,
and Section~\ref{sec:considerations:materialization} on physical design of cubes.

\subsection{Considerations for cube crawling}
\label{sec:considerations:cube-crawling}
We consider several algorithmic aspects for crawling:
\textbf{(1) Early stopping during region exploration}.
We note that the na\"{i}ve cube crawling algorithm,
Algorithm~\ref{alg:theoretical-crawling},
is impractical as the region space may grow exponentially in the number of dimensions to crawl.
For this we observe the natural lattice structure on the region space,
and formulate a generalized apriori property of model signals (Definition~\ref{def:apriori}).
This gives a natural top-down region exploration that allows early stopping (Algorithm~\ref{alg:topdown-crawling}).
\textbf{(2) Model pushdown.}
Another important algorithmic opportunity is that
we can push some model computation into the cube processing,
and filter regions without sending the extracted cube features to the model.
This can give significant performance benefits in a distributed implementation,
where cube processing and model evaluation are on different servers.
\textbf{(3) Top-N filtering}. A common filtering strategy is to compute the
top $N$ regions w.r.t. some signals. A na\"{i}ve strategy is to check all
regions to find the top $N$. To this end, we discuss how can we leverage
ideas in~\cite{FLN02} to enable optimizations (1) and (2) during Top-N computation,
so that one doesn't have to check all regions.

\begin{algorithm}[htb]
  \small
  \caption{\small$\TopDownCubeCrawl$: A more practical top-down algorithm paradigm
    for cube crawling. This paradigm allows early stopping, by having the crawler interact with signals extracted by models from a region.}
  \begin{algorithmic}[1]
    \REQUIRE $\bfC$: $\sf \color{blue} AbstractCube$;
      {\bf region\_space}: A specification of region space;
      {\bf models}: $\sf \color{blue} List[Model]$
    \ENSURE $\bfR$: $\sf \color{blue}  AbstractCube$.
      A cube mapping regions to signals. For simplicity,
      here the output cube is encoded as a set of
      region-signals pairs.
    \STATE $\bfR = \emptyset$
    \STATE Initialize pending\_regions as only contain the empty region
    \WHILE{pending\_regions is not empty}
      \STATE region = pending\_regions.pop()
      \STATE subcube = $\bfC$[region]
      \STATE {\color{green!90} \# Evaluate models on a region.}
      \STATE $r$ = $\emptyset$
      \FOR{$f$ \textbf{in} models}
        \STATE {\color{green!90} \# Extract features from the region subcube}
        \STATE data\_frame = subcube($f$.attribute\_features, $f$.metric\_features)
        \STATE {\color{green!90} \# Apply the model and append results}
        \STATE $r$.append($f$(data\_frame))
      \ENDFOR
      \STATE {\color{green!90} \# Early stopping!}
      \IF{$r$ fails apriori pruning}
        \STATE continue
      \ENDIF
      \STATE {\color{green!90} \# Append more regions to crawl if we don’t stop early}
      \STATE pending\_regions.append(region.children)
      \STATE {\color{green!90} \# Collect $r$ into $\bfR$}
      \IF{signals $r$ passes all thresholds}
        \STATE $\bfR$.append((region, $r$))
      \ENDIF
    \ENDWHILE
    \RETURN $\bfR$.
  \end{algorithmic}
  \label{alg:topdown-crawling}
\end{algorithm}

\textbf{Top-down crawling with early stopping}.
To begin with, we note that a region space may grow exponentially as the number of dimensions increases.
For example, if we have $n$ dimensions to crawl, there are $2^n$ subsets of dimensions
(each subset is called \emph{grouping set}), and each grouping set can give a large number of regions.
Therefore, it is important to devise strategies to prune the search space w.r.t. the specified thresholds. 
For this, we note a natural {\em partial order} among {\em regions}:
For two regions $\gamma_1$ and $\gamma_2$,
$\gamma_1 \prec \gamma_2$, if every dimension value
in $\gamma_2$ is also in $\gamma_1$. In other words,
$\gamma_1$ is more ``fine-grained'',
while $\gamma_2$ is more ``coarse-grained''.
(e.g., (State=CA, City=MTV) $\prec$ (State=CA))
Definition~\ref{def:apriori} defines Apriori property of model signals:
\begin{definition}[\textbf{Apriori property of model signals}]
  \label{def:apriori}
  Let $\sigma$ be a signal computed by a model (e.g., $\sf support$
  in the FrequentItemsetModel). The signal is said to satisfy the apriori property
  if for any two regions $\gamma_1 \prec \gamma_2$,
  $\sigma[\gamma_1] \prec \sigma[\gamma_2]$, where $\sigma[\gamma]$ means the
  value of signal $\sigma$ on region $\gamma$.
\end{definition}
There are many examples of apriori model signals. For example, if a signal is nonnegative and admits $\sf SUM$ aggregation,
then it satisfies the Apriori property. Another example is $\sf COUNT\_DISTINCT$, which also satisfies the Apriori property,
because more fine-grained regions can only have less distinct counts.
With this definition, if we crawl with models that computes signals that satisfy the above Apriori property,
then we can early stop if a coarse-grained region already fails some thresholding.
This thus gives rise to top-down crawling as sketched in Algorithm~\ref{alg:topdown-crawling}.
We note the following:

\textbf{\em Region tree}. At line 19, one needs to define `region.children',
  which gives a tree of regions for exploration. We note that finding a good tree
  will be important for performance. For example, one may start by considering
  dimensions that have low cardinality, and then higher-cardinality dimensions,
  so we go from coarser regions to more fine-grained regions. Another factor
  that may have a significant impact on region tree is \emph{data topology}:
  For example, if there is a hierarchy of dimensions
  (such as Country > State > City),
  then one only needs to crawl degree-1 regions
  (degree of a region is defined in preliminaries) in Country,
  degree-2 regions in [Country, State], and degree-3 regions in [Country, State, City].

\textbf{\em Exploration strategy}. At line 3, different data structures
  can be used for the `pending\_regions' variable,
  and thus give different exploration order of regions.
  For example, instantiating it as a stack gives depth first exploration,
  and as a queue gives breadth first exploration.
  Note the ``embarrassing parallelism'' for regions in `pending\_regions',
  and one can explore regions in it in parallel.

\textbf{Model pushdown to cube processing}.
In a distributed implementation, one may implement the cube function using using a distributed query engine.
In that case, the evaluation of models over the extracted views
and the cube processing may happen on different servers.
To reduce data transfer, we can push some model computation to the cube-data processing
if they are easily expressible in SQL,
so the signal computation is completed already during cube function computation.
This pushdown becomes particularly useful if there is a threshold filtering on the signal,
so we may directly filter a region without transferring data to model evaluation servers.

\textbf{Top-N cube crawling}.
While thresholding on the model signals is a most intuitive method, sometimes it
can get difficult in determining what the thresholds should be. In such cases,
a next intuitive filtering method is to compute the Top-N results with respect to
some model signals. An interesting algorithmic consideration here is how to avoid
\emph{examining every region} in this computation. In view of 
Algorithm~\ref{alg:topdown-crawling}, this requires us to leverage apriori pruning
on model signals, which in turn needs a threshold on the signals, which do not
exist for the Top-N computation.
To this end, suppose we are working with a single model signal $\nu$ which satisfies
the apriori property. We adapt the basic idea of ``dynamic thresholding'' from \cite{FLN02}:
\emph{If our goal is to find top $n$ regions,
then if we have seen any set $\cal F$ of regions of size larger than $n$,
then the $(n+1)$-th region in $\cal F$ can be used as threshold on $\nu$}.
While this is simple to implement in the central case using a priority queue,
it becomes challenging in a distributed implementation:
First, maintaining a centralized priority queue may incur much coordination.
Second, there is a now a tradeoff between examining a large batch of regions (in parallel)
with a loose threshold, versus examining regions in smaller batches but with better thresholds.
To this end, we now adopt a simple design where we examine once some low-degree regions,
compute thresholds from them, and then use those signals for thresholding.

\subsection{Cube join: Local vs. Global}
\label{sec:considerations:join}

As we discussed in the definition of cube join, there are two ways to perform cube join:
One is to first load views locally from two cubes, and join them as tables,
and the other is to join the cellsets of two cubes, which could then induce any
view of the joined cube. These two approaches have different tradeoffs when one
performs cube analysis on a joined cube. For example, for a cube crawling on a joined
cube, the local approach gives a lazy way to join data since we only join data
if there is a need to explore a certain view. However, this means that there will
potentially be a huge amount of small joins of tables, which may incur significant
data processing overhead by having too many join queries.
On the other hand, the global approach prepares the entire cellset of the joined cube,
which may incur significant overhead in the beginning, but it gives only one join request,
and any subsequent view can be efficiently computed from the joined cellset.
The above discussion indicates that there can be a third possibility of
a hybrid approach, where we partially perform local join,
and partially perform global join.

\subsection{Physical design of cubes}
\label{sec:considerations:materialization}
While logically cube is a function from RegionFeatures space to relational tables,
its physical design (i.e., how do we encode this function)
need careful considerations in order to benefit HoCA operators' performance.
One basic consideration is that in many cases,
it will be too costly to crawl the tables if one has to aggregate slices
on-the-fly from some base table. Therefore, it is useful to provide encoding
where the cube is (partially) materialized.

\textbf{Materialization granularity}.
A naive strategy is to simply materialize the views needed in a crawling.
For example, suppose we are doing a daily continual monitoring,
where we want to analyze a timeseries of seven-days data for each day.
Then na\"{i}vely, we materialize a space of tables, each of which is a seven-day timeseries.
However, such materialization strategy is rather wasteful:
For example for any two consecutive days,
the materialization has an overlap of six days of aggregated data.
A much better strategy in this case is to organize the materialized data by date:
For each day, we materialize the daily data of different regions into a chunk and store them.

\textbf{Re-chunking}.
Suppose that we have an encoding of cube organized data by date, as discussed above.
When we actually do the cube crawling, we need to stitch together 7 days of data.
For performance, we then would like to \emph{re-chunk} the cube so that for the
same region, the 7-day data are stitched together into a slice.
Note that the re-chunked cube and the cube above are two different encoding of the
same logical cube (i.e., RegionFeatures space to tables function).

%% file: impl-eval.tex
We have implemented the current HoCA framework using a distributed SQL engine~\cite{F1Query},
as well as Apache Beam~\cite{PythonBeam}.
We give some important lessons of our implementation of cube crawling.

$\bullet$ For our pure SQL engine implementation,
we leveraged TableValuedFunction (TVF) framework to execute user programmed RAMs.
At the high level, a coordinator receives a crawling request,
and factors the crawling request into SQL-TVF queries
(the exact factoring is complicated, which depends on the exploration strategy):
The SQL part is to prepare region-features tables,
and the TVF part is for evaluating RAMs on the materialized tables.
We term this implementation of crawler the \textbf{SqlTvf crawler}.

$\bullet$ One lesson we learned is that the current TVF framework implemented in our query engine
does not scale well as the cube grows: A key reason is that the current TVF scheduling is not designed
with massive slice model evaluation (from crawling) in mind. We are currently exploring a better
TVF scheduling design that would overcome this limitation.

$\bullet$ To mitigate the above scalability issue, we turn to use Apache Beam for model evaluation.
Basically, we only rely on SQL to generate tables induced by region-features, and materialize these
tables into some storage, and then we invoke a carefully designed Beam pipeline to evaluate models
on these slices. Since Beam provides more flexible programming capabilities, we can capture
more easily optimizations such as Apriori pruning described earlier. These optimizations,
together with powerful scheduling of our internal Beam runtime (Flume C++), scales really well on
large cubes. We term this implementation of crawler the \textbf{SqlBeam crawler}.
Note that SqlBeam crawler has an additional materialization cost, which hurts latency.
Currently we use the SqlTvf crawler for streaming use cases, and use the
SqlBeam crawler for large-scale batch use cases.

\textbf{Evaluation goal}.
We evaluate our SqlTvf crawler and SqlBeam crawler.
At the high level we examine two questions:
(1) {\em Latency: What are the characteristics of latency of the crawlers?}
(2) {\em Scalability: What are the characteristics of scalability of the crawlers?}

\begin{table}[htb]
  \centering\small
  \begin{tabular}{l|l}
      \toprule
      ShiftModel &
      \thead[l]{
        A simple model that analyzes change of a summable metric \\
        (e.g., Revenue). A nontrivial part is that we normalize \\
        using the population change, because RAM allows one to \\
        access population features. This normalization is \\
        very useful for filtering.}\\\hline
     CausalImpactModel &
     \thead[l]{
       A timeseries analysis model based on~\cite{BGKRS15}. Our \\
       implementation added more analysis leveraging the population \\
       timeseries, which is useful in the holistic analysis scenario.}\\\hline
     EntityWeightModel &
     \thead[l]{A simple model for analyzing weights of multiple \\
       entities. In our experiments, this model is used to filter \\
       ``small'' timeseries.}\\
      \bottomrule
    \end{tabular}
  \caption[]{\small Summary of models used in experiments.}
  \label{tab:expt-models}
\end{table}

{\bf Models and analyses}.
Table~\ref{tab:expt-models} summarizes the models we used in experiments.
With these models, we evaluate two analyses.

\textbf{\em Shift analysis}.
For this analysis, for each dataset we chose a specific date to create a test-control split;
the control group contains all data points before the date Sep. 1, 2021 and
the test group contains all the other data points.
During the analysis, we use the early stopping with the Apriori property
(Definition~\ref{def:apriori}) of the apriori support signal.
We set a threshold of $0.075$ as an early stopping criterion.
    
\textbf{\em Time series anomaly detection}.
For this analysis, we combine two RAMs, EntityWeightModel and CausalImpactModel,
which gives an instance of multi-model crawling described in Section~\ref{sec:multimodel}.
Specifically, each region has a weight, the sum of the values in the weight column over a specified time frame.
The EntityWeightModel first filters regions so only regions with large enough weight remain,
then the CausalImpactModel performs timeseries analysis on these regions to detect outliers.
We apply this analysis daily in a continual fashion: For each day, we analyze today's observation
in a 30d window, and determine whether the observation of today is an outlier,
In the experiments, we run CausalImpactModel for the data point on a specific date, Oct. 1, 2021, compared to the data points from the 30 days in September 2021.
For EntityWeightModel filtering, we use the metric columns of the datasets as the weight columns.
Since the metric features of two datasets show different scale, we set different threshold to filter the evaluation result.
For the SqlTvf crawler experiments, the threshold for the COVID-19 Open Data is 18,000,
and the threshold for the CMS Open Payments is 5,500,000.

\textbf{Datasets}.
We evaluate the distributed crawling performance on two real-world datasets. Table~\ref{tab:dataset-summary} summarizes those two datasets.

\textbf{\em COVID-19 Open Data}.
COVID-19 Open Data\footnote{\url{https://health.google.com/covid-19/open-data/}} is Google's public data collection that contains up-to-date COVID-19-related information. COVID-19 Open Data consists of data from different sources, including demographics, economy, epidemiology, government response, weather, etc. For our evaluation of distributed crawling performance, we use one metric feature -- the number of new confirmed cases -- and 22 attribute columns, including age group, region, weather and government responses.

\textbf{\em Center for Medicare \& Medicaid Services (CMS) Open Payments}.
CMS Open Payments\footnote{\url{https://www.cms.gov/OpenPayments/Data}} is a public data set about payments and transfers of value from reporting entity, such as pharmaceutical manufacturers, to covered recipients, such as physicians. Starting from the program year 2013, CMS publishes Open Payments data annually. We merged nine annual data, from the program year 2016 to the program year 2021, into one table. We use one metric feature -- the total amount of payment in US dollars -- and 22 attribute columns, which encode information about
recipients and payments.

\begin{table}[htb]
  \centering\small
  \begin{tabular}{l|c|c|c}
      \toprule
      \textbf{Dataset} & \textbf{\# rows} & \textbf{\# attributes} & \textbf{\# regions of degree $\le$ 3} \\ \hline\hline
      COVID-19 & 21.2M & 22 & 16.8M \\\hline
      CMS & 55.8M & 22 & 607.6M \\\hline
      \bottomrule
    \end{tabular}
  \caption[]{\small Summary of datasets used for the evaluation of distributed crawling performance.}
  \label{tab:dataset-summary}
\end{table}

\textbf{End-to-end performance evaluation}.
\textbf{\em Latency}. When measuring the SqlTvf crawler latency, we used 10 TVF workers to handle each crawling request. We made 200 crawling requests for each dataset-analysis pair to compute the average performance. Figure~\ref{fig:latency} presents the measured latency -- Figure~\ref{fig:latency-covid19} shows the latency on COVID-19 Open Data, and Figure~\ref{fig:latency-cms} shows the latency on CMS Open Payments. To summarize the result, on the COVID-19 Open Data, the average latency is 192.25 seconds for shift analysis and 353.96 seconds for anomaly detection. On the CMS Open Payments, the average latencies are 218.75 seconds for shift analysis and 630.24 seconds for anomaly detection. This shows that the SqlTvf crawler generally achieves low average latency for the example use cases. 

\begin{figure}[htb]
     \centering
     \begin{subfigure}[b]{0.45\textwidth}
         \centering
         \includegraphics[width=\textwidth]{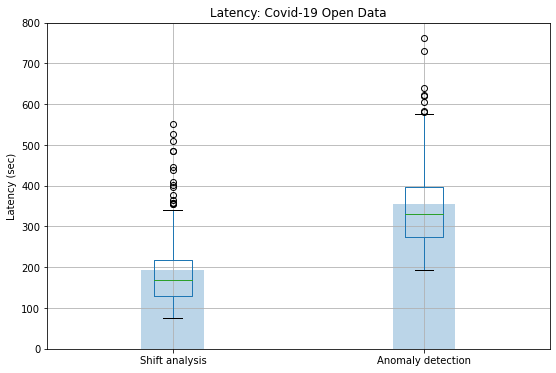}
         \caption{COVID-19 Open Data}
         \label{fig:latency-covid19}
     \end{subfigure}
     \begin{subfigure}[b]{0.45\textwidth}
         \centering
         \includegraphics[width=\textwidth]{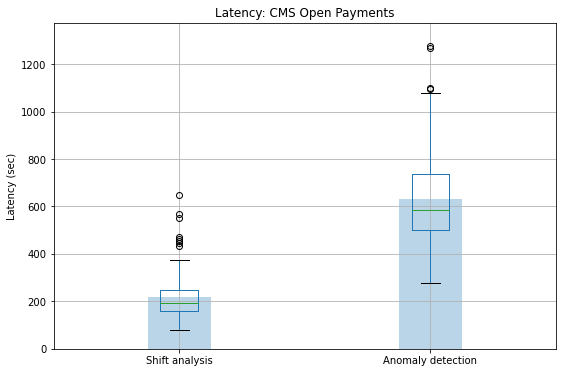}
         \caption{CMS Open Payments}
         \label{fig:latency-cms}
     \end{subfigure}
     \caption{\small Latency of the SqlTvf crawler implementation. Each column presents the average latency (bar), the median (green line), and the distribution of the latencies (box plot).}
     \label{fig:latency}
\end{figure}

\begin{figure}[htb]
     \centering
     \begin{subfigure}[b]{0.45\textwidth}
         \centering
         \includegraphics[width=\textwidth]{
         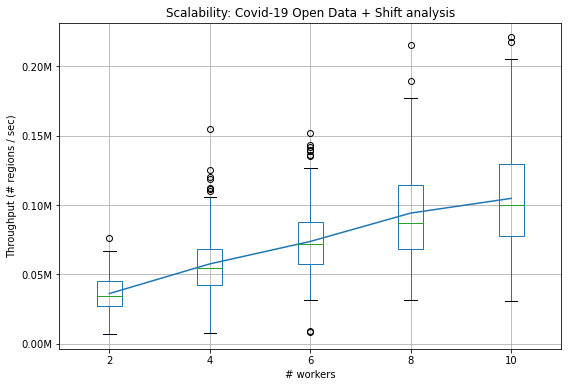}
         \caption{COVID-19 + Shift}
         \label{fig:scalability-covid19-shift}
     \end{subfigure}
     \hspace*{\fill}
     \begin{subfigure}[b]{0.45\textwidth}
         \centering
         \includegraphics[width=\textwidth]{
         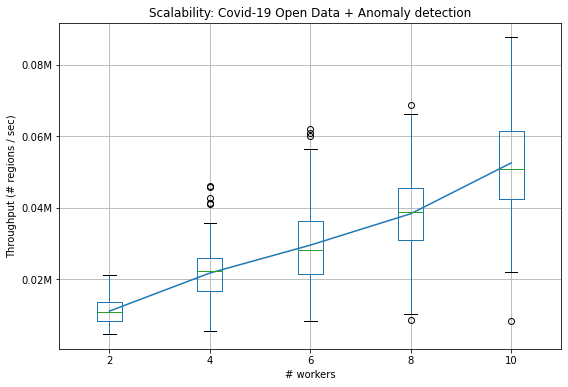}
         \caption{COVID-19 + Anomaly}
         \label{fig:scalability-covid19-anomaly}
     \end{subfigure}\\
     \begin{subfigure}[b]{0.45\textwidth}
         \centering
         \includegraphics[width=\textwidth]{
         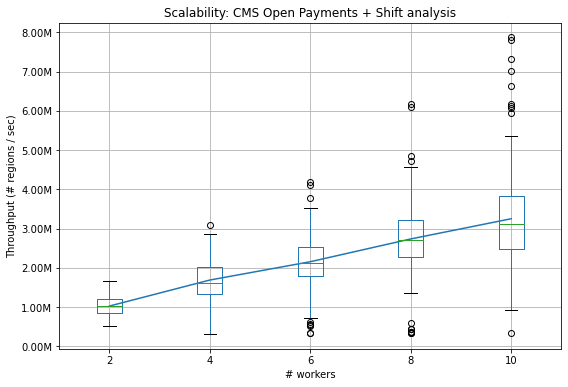}
         \caption{CMS + Shift}
         \label{fig:scalability-cms-shift}
     \end{subfigure}
     \hspace*{\fill}
     \begin{subfigure}[b]{0.45\textwidth}
         \centering
         \includegraphics[width=\textwidth]{
         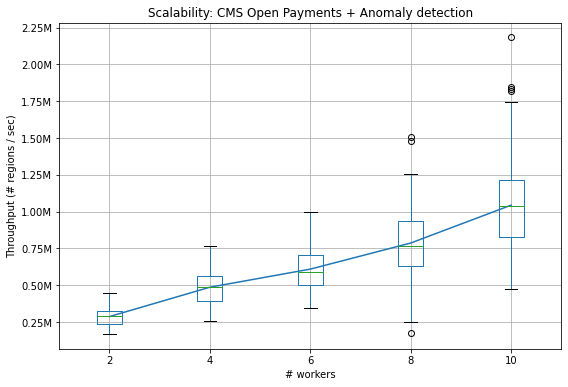}
         \caption{CMS + Anomaly}
         \label{fig:scalability-cms-anomaly}
     \end{subfigure}
     \caption{\small Scalability of the SqlTvf crawler implementation.}
     \label{fig:scalability}
\end{figure}

\textbf{\em SqlTvf crawler scalability}.
To evaluate SqlTvf crawler scalability, we first measured the latency of 200 requests on varying TVF worker sizes: 2, 4, 6, 8, and 10.
Then we computed the throughput by dividing the number of regions by each measured latency.
Figure~\ref{fig:scalability} summarizes how the throughput increases as the number of TVF workers grows.
In summary, both analyses show the increase of throughput in the range of TVF worker sizes on all the dataset-analysis pairs.
However, the demonstrated scalability of the SqlTvf crawler is limited to small workload.
For example, during the anomaly detection on COVID-19 Open Data, the TVF framework evaluates on 821 regions after the filtering by EntityWeightModel.
If we increase the number of after-filtering regions, the SqlTvf crawler suffers from the TVF scheduling problem and the latency increases significantly as a result.
To overcome this limitation, we implemented the SqlBeam crawler that a Beam pipeline evaluates RAMs.

\textbf{\em SqlBeam crawler scalability}.
For SqlBeam crawler, we vary the workload (from small to big), and measure its performance.
For 11 different filtering threshold, we measured the E2E latency (of 50 requests) for anomaly detection on COVID-19 Open Data.
Figure~\ref{fig:latency-f1beam} presents the measured E2E latency of the SqlBeam crawler.
This shows that the SqlBeam crawler can scale to large workload (e.g., 247202 regions),
and the average E2E latency only gets doubled while the workload grows more than 100 times.

\begin{figure}[htb]
  \centering
  \begin{subfigure}[b]{0.75\textwidth}
    \centering
    \includegraphics[width=\textwidth]{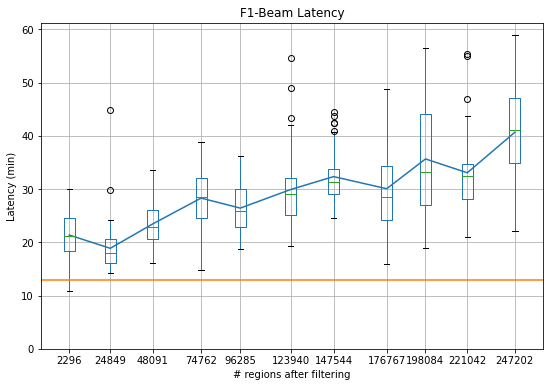}
  \end{subfigure}
  \caption{\small Latency of the SqlBeam crawler implementation. The boxplots present the latency distributions.
  The line plots show the change in average latency (blue line) and the average latency of the materialization cost
  and the filtering (orange line).}
  \label{fig:latency-f1beam}
\end{figure}

%% file: discussion.tex
\textbf{Applications of HoCA}.
Even in this early stage, our HoCA offering has attracted more than 30 teams
to build data-insights systems with it. We identify three kinds of main
application areas:
\textbf{(1) System monitoring}.
Many of our partner teams are working with extremely complex systems,
and monitoring the health of these systems is critical.
Our outlier detection construction, and its variants of slow-bleed detection
(by switching to a trend-change detection model), and recurrent cube crawling,
have found major applications for these use cases:
HoCA analysis enabled much faster and more accurate detection of precise issues
(such as which version is having anomalous behaviors), and thus prevents
significant revenue loss.
\textbf{(2) Business intelligence (BI)}.
HoCA constructions such as metric attribution and correlation detection
have been applied in several BI systems to help analysts understand
what is going on and identify new business opportunities.
For example, by composing several HoCA operators, we can quickly identify
that certain account has a revenue anomaly, and is strongly correlated with
budget allocation of some campaigns.
Some of these applications have generated significant revenue uplift.
\textbf{(3) Experimentation analysis}.
Several major internal experimentation analysis systems requires
slicing the data and do A/B testing on the slices to understand fine-grained
impact of an experiment. These core analytic functionalities can be captured by
cube crawling by programming appropriate RAMs (however, these RAMs tend to
have more statistical flavor since the input data is usually a sample).
Migrating these systems to use HoCA has significantly improved their latency
and scalability, and sometimes enabled brand new functionalities
(e.g., novel analysis leveraging HoCA abstractions).

\textbf{More HoCA operators}.
Section~\ref{sec:emulation} shows that cube crawling can emulate DIFF~\cite{AKSGXSASM21},
which in turn can emulate a variety of data explanation engines (see references therein).
To this end, if we examines literature with HoCA in mind,
there are several work that can be interpreted as HoCA operators.
For example, operators described in Multistructural Database~\cite{FaginGKNST05}
and Cascading Analyst~\cite{RSY18} can both be viewed as HoCA operators,
though they have much more adhoc semantics compared to cube crawling.

\textbf{Beyond operators}.
As we have demonstrated in recurrent cube crawling (Section~\ref{sec:recurrent-crawling}),
there can be complex HoCA pipelines with cubes (explicitly in the program) flowing in between.
These pipelines require complex cube management,
and currently have esoteric constructions. Therefore, it is of importance to research
how can we simplify and systematize these pipeline constructions
(which we call HoCA pipelines).
Another intriguing direction is whether HoCA cubes can be useful for \emph{learning}
(which we call the \emph{cube learning}).
To this end, one direction we are exploring is \emph{cube forecasting}:
Classical forecasting usually assumes that the input data is one slice,
yet in realistic scenario, input data can often be modeled as a cube
where the regions encode useful structural information, and conceivably one can
leverage this structure to achieve better timeseries forecasting.
We believe there is more since HoCA cubes provide a rich data model for
data analysis and learning.

%% file: conclusion.tex
In this paper we described Holistic Cube Analysis (HoCA),
an approach to augment the capabilities of relational queries for
\emph{data insights}.
At the foundational level, HoCA defines AbstractCube,
a data type defined as a \emph{function} from RegionFeatures space
to relational tables,
which allows one to go from a set of \emph{tuples} (i.e., a table)
to a set of \emph{tables} (i.e., a set of set of tuples).
As a result, queries over cubes become cube-to-cube transformations
(i.e., higher-order functions).
We described two first operators,
cube crawling and cube join, in the current HoCA framework.
The power of these two operators comes from
\emph{programmable RAMs},
\emph{multi-model crawling},
and the \emph{composition of operators}
(between cube crawling and cube join, and with other relational queries).
HoCA has found many useful data-insights applications,
across different areas such as system monitoring, experimentation analysis, and business intelligence.
There are various future avenues for HoCA,
such as efficient implementation of HoCA operators,
more useful HoCA operators, HoCA pipelines, HoCA and interactive analysis,
and last but not least, exploring connections with machine learning.

%% file: aumann-shapley/density-aumann-shapley.tex
A \emph{density metric} $\rho$ is a ratio metric 
where both the numerator and the denominator admit sum aggregation
(i.e., it is ``SUM/SUM'').
We denote the numerator of the density by $w$, and denominator of the density by $s$.
Let $\bfC$ be a cube. Suppose that we have a dimension `is\_test' which splits the data
into test (`is\_test'=True) and control (`is\_test'=False).
Therefore, at the population level, we have population metrics
for test ($w^p_t$, $s^p_t$) and control ($w^p_c$ and $s^p_c$),
which induces a difference:
\begin{align}
  \label{formula:population-ratio-diff}
  C_\rho = \frac{w^p_t}{s^p_t} - \frac{w^p_c}{s^p_c}
\end{align}

Let $r$ be a region in the cube. At the region level, we have region metrics
for test ($w^r_t$ and $s^r_t$) and for control ($w^r_c$ and $s^r_c$),
which induces metric changes
$\Delta w^r = w_t^r - w_c^r$ and $\Delta s^r = s_t^r - s_c^r$.
We consider the following question:
\begin{center}
  \emph{Can we examine the region level metric change and the population level metric change,
  and compute a real value to assign to the region, to ``model'' the contribution of the region
  to the population level density change?}
\end{center}

Note that the answer is not trivial because the ratio nature of the density metric.
The answer to the question can be implemented as a region analysis model for cube crawling
to help find regions that contribute significantly to the change of a density metric at the top level.
We give a principled answer to the question based on the Aumann-Shapley method.
To apply the Aumann-Shapley method, we consider the following function:
\begin{align}
  F(w, s, w', s') = \frac{w + w'}{s + s'}
\end{align}
where we should think of $w,s$ are mass and volume in a region,
and $w',s'$ as mass and volume
in the complement of the region (i.e., outside of the region).
Importantly, $F$ gives a ``region-population'' modeling for attribution: For each region,
we only need the metric values of the region, and of the population, to compute the four
values $w, s, w', s'$.

With the above, then the control (base) point for Aumann-Shapley method is exactly
$P_0 = (w_c^r, s_c^r, w_c^p-w_c^r, s_c^p-s_c^r)$,
and the test (comparison) point is exactly
$P_1 = (w_t^r, s_t^r, w_t^p-w_t^r, s_t^p-s_t^r)$.
Our goal is thus to compute attribution for the first and second
coordinates (where the first coordinate changes from $w_c^r$ to $w_t^r$,
and the second coordinate changes from $s_c^r$ to $s_t^r$).
To apply the Aumann-Shapley method, we form a path:
\begin{align}
  \label{eq:path}
  h(t) = (1-t) \cdot P_0 + t \cdot P_1 \qquad(0 \le t \le 1)
\end{align}
Let $g(t) = F(h(t))$, therefore we have (this is the gist of Aumann-Shapley method)
\begin{align}
  \label{eq:AS-path-integral}
  F(P_1) - F(P_0)
  = g(1) - g(0)
  = \int_0^1 g'(t) dt
  = \int_0^1 \sum_i
  \frac{\partial F}{\partial z_i}\frac{\partial z_i}{\partial t}dt
  = \sum_i \left(\int_0^1
  \frac{\partial F}{\partial z_i}\frac{\partial z_i}{\partial t}dt\right)
\end{align}
where $z_i = h(t)_i$ ($i=1,2,3,4$) are the variables of $F$,
that is, parameterizations:
\begin{itemize}
\item $w(t) = z_1 = h(t)_1 = (1-t)w_c^r + tw_t^r$.
\item $s(t) = z_2 = h(t)_2 = (1-t)s_c^r + ts_t^r$.
\item $w'(t) = z_3 = h(t)_3 = (1-t)(w_c^p-w_c^r) + t(w_t^p-w_t^r)$.
\item $s'(t) = z_4 = h(t)_4 = (1-t)(s_c^p-s_c^r) + t(s_t^p-s_t^r)$.
\end{itemize}

\subsection{Attribution of numerator}

Therefore the attribution to the first coordinate $z_1$, is simply
\begin{align*}
  \int_0^1\frac{\partial F}{\partial z_1}\frac{\partial z_1}{\partial t}dt
\end{align*}
We note that
\begin{align*}
  \frac{\partial z_1}{\partial t} = w_t^r - w_c^r
\end{align*}
and
\begin{align*}
  \frac{\partial F}{\partial z_1} = \frac{1}{s(t)+s'(t)}
  = \frac{1}{(1-t)s_c^p + t s_t^p}
  = \frac{1}{(s_t^p-s_c^p)t + s_c^p}
\end{align*}
Therefore the attribution of the first coordinate is
\begin{align}
  \label{eq:numerator-attribution}
  \begin{split}
    \int_0^1\frac{w_t^r-w_c^r}{s(t)+s'(t)}dt
    = &(w_t^r-w_c^r)\int_0^1\frac{1}{(s_t^p-s_c^p)t + s_c^p}dt \\
    = &\frac{w_t^r-w_c^r}{s_t^p-s_c^p}
    \ln\Big((s_t^p-s_c^p)t + s_c^p\Big)\Big\rvert_0^1 \\
    = &(w_t^r-w_c^r)\frac{\ln(s_t^p) - \ln(s_c^p)}{s_t^p-s_c^p}
  \end{split}
\end{align}

\subsection{Attribution of denominator}
Now we turn to the attribution of the second coordinate, which is
\begin{align*}
  \int_0^1\frac{\partial F}{\partial z_2}\frac{\partial z_2}{\partial t}dt
\end{align*}
where
\begin{align*}
  \frac{\partial z_2}{\partial t} = s_t^r - s_c^r
\end{align*}
and
\begin{align}
  \label{eq:1}
  \frac{\partial F}{\partial z_2} = -\frac{w(t) + w'(t)}{(s(t) + s'(t))^2}
  = -\frac{(1-t)w_c^p + tw_t^p}{\Big((1-t)s_c^p+ts_t^p\Big)^2}
  = -\frac{(w_t^p-w_c^p)t+w_c^p}{\Big((s_t^p-s_c^p)t + s_c^p\Big)^2}
\end{align}
Let $f(t) = \frac{1}{(s_t^p-s_c^p)\Big((s_t^p-s_c^p)t + s_c^p\Big)}$,
$g(t) = (w_t^p-w_c^p)t+w_c^p$, then
\begin{align*}
  (\ref{eq:1}) = & \int_0^1 f'g dt \\
  = & fg\Big\rvert_0^1 - \int_0^1fg'dt\qquad\qquad\text{(Integration by parts)}\\
  = & \frac{1}{s_t^p-s_c^p}\Big(\frac{w_t^p}{s_t^p} - \frac{w_c^p}{s_c^p}\Big)
      - \frac{w_t^p-w_c^p}{s_t^p-s_c^p}\int_0^1\frac{1}{(s_t^p-s_c^p)t+s_c^p}dt\\
  = & \frac{1}{s_t^p-s_c^p}\Big(\frac{w_t^p}{s_t^p} - \frac{w_c^p}{s_c^p}\Big)
      - \frac{w_t^p-w_c^p}{(s_t^p-s_c^p)^2}
      \ln\Big((s_t^p-s_c^p)t+s_c^p\Big)\Big\rvert_0^1 \\
  = & \frac{1}{s_t^p-s_c^p}\Big(\frac{w_t^p}{s_t^p} - \frac{w_c^p}{s_c^p}\Big)
      - \frac{w_t^p-w_c^p}{(s_t^p-s_c^p)^2}\Big(\ln(s_t^p) - \ln(s_c^p)\Big)
\end{align*}
Therefore we get attribution of denominator:
\begin{align}
  \label{eq:denominator-attribution}
  \begin{split}
    (s_t^r-s_c^r)\bigg\{
    \frac{1}{s_t^p-s_c^p}\Big(\frac{w_t^p}{s_t^p} - \frac{w_c^p}{s_c^p}\Big)
    - \frac{w_t^p-w_c^p}{(s_t^p-s_c^p)^2}\Big(\ln(s_t^p) - \ln(s_c^p)\Big)
    \bigg\}
  \end{split}
\end{align}

\subsection{Combining everything together}

The theory of Aumann-Shapley attribution says that the attribution of
$z_1$ and $z_2$ is the sum of their individual attributions, this thus
gives
\begin{align}
  \label{eq:region-attribution}
  \begin{split}
    (\ref{eq:numerator-attribution}) + (\ref{eq:denominator-attribution})
    = & (w_t^r-w_c^r)\frac{\ln(s_t^p) - \ln(s_c^p)}{s_t^p-s_c^p} +
        (s_t^r-s_c^r)\bigg\{
        \frac{1}{s_t^p-s_c^p}\Big(\frac{w_t^p}{s_t^p} - \frac{w_c^p}{s_c^p}\Big)
        - \frac{w_t^p-w_c^p}{(s_t^p-s_c^p)^2}\Big(\ln(s_t^p) - \ln(s_c^p)\Big)
        \bigg\}\\
        = & \frac{(w_t^r-w_c^r)(s_t^p-s_c^p) -
          (w_t^p-w_c^p)(s_t^r-s_c^r)}{(s_t^p-s_c^p)^2}
        \Big(\ln(s_t^p) - \ln(s_c^p)\Big) +
        \frac{s_t^r-s_c^r}{s_t^p-s_c^p}
        \Big(\frac{w_t^p}{s_t^p} - \frac{w_c^p}{s_c^p}\Big)
  \end{split}
\end{align}

\subsection{Completeness property of our region attributions}

Note that we have used a particular way to perform the attribution: For every single region,
we form a function to model the region-population relationship and by analyzing that function,
we compute a value as attribution for the region; and we repeat this for all regions.
We now show that our method has some nice completeness properties.
To see this, we will denote by $\ras(\gamma)$ the attribution computed by
our algorithm for region $\gamma$ (RAS stands for ``Region Aumann-Shapley'').
We observe that the population level metrics are constants,
so the attribution formulas can be rewritten as follows:

\begin{align}
  \label{formula:AS-simplified}
  \bigg[(s_t^p-s_c^p)(w_t^r - w_c^r) - (w_t^p-w_c^p)(s_t^r-s_c^r)\bigg] \cdot K 
  + \left(\frac{s_t^r-s_c^r}{s_t^p-s_c^p}\right) \cdot C_\rho,
\end{align}

where $K$ is some constant only depending on the population, which we don't care in the following;
and $C_\rho$ is actually exactly the \emph{population density change}.
Now, we observe the following ``summable'' property:

\begin{lemma}[\textbf{Region attribution is additive}]
  Consider two regions $\beta, \gamma$ that are disjoint (namely the data points involved are disjoint,
  which we write as $\beta \cap \gamma = \emptyset$). Then
  $\ras(\beta \cup \gamma) = \ras(\beta) + \ras(\gamma)$.
\end{lemma}
\begin{proof}
  At the region level, we have the metrics for $\beta, \gamma$, which is:
  $\beta: (w^\beta_c, s^\beta_c) \rightarrow (w^\beta_t, s^\beta_t)$, and
  $\gamma: (w^\gamma_c, s^\gamma_c) \rightarrow (w^\gamma_t, s^\gamma_t)$.
  Because $\beta, \gamma$ are disjoint, therefore for the union we have:

  $$\beta \cup \gamma: (w^\beta_c+w^\gamma_c, s^\beta_c+s^\gamma_c)
  \rightarrow (w^\beta_t+w^\gamma_t, s^\beta_t+s^\gamma_t).$$

  So we can compute $\ras(\beta \cup \gamma)$ using formula (\ref{formula:AS-simplified}),
  and one can see that it is just sum of $\ras(\beta)$ and $\ras(\gamma)$.
\end{proof}

Note that we may not be able to represent $\beta \cup \gamma$ as a region,
so the conclusion of this lemma is quite nontrivial.
In fact, it immediately implies the following two:

\begin{corollary}[\textbf{Global population completeness}]
  Suppose that we have a collection $\cal F$ of \emph{disjoint} regions,
  where these regions ``sum up'' to the population. Then
  \begin{align*}
      \sum_{r \in {\cal F}}\ras(r) = C_\rho,
  \end{align*}
  where $C_\rho$ is the top level density change.
\end{corollary}

\begin{proof}
  This immediately follows from the ``summable'' result above. Here we give a alternate direct proof.
  Summing over the attributions of regions $r \in {\cal F}$, which gives:
  \begin{align*}
    \sum_{r \in {\cal F}} \left(\bigg[(s_t^p-s_c^p)(w_t^r - w_c^r) - (w_t^p-w_c^p)(s_t^r-s_c^r)\bigg] \cdot K 
    + \left(\frac{s_t^r-s_c^r}{s_t^p-s_c^p}\right) \cdot C_\rho\right),
  \end{align*}
  We observe that this gives exactly $C_\rho$, because the first summand gives rise to a sum that is zero,
  while for the second summand, $\sum_r (s_t^r-s_c^r)/(s_t^p-s_c^p) = 1$ (exercise).
  This gives global population completeness.
\end{proof}

\begin{corollary}[\textbf{Subpopulation completeness}]
  Suppose we have a region $\gamma$ which has density change $C^\gamma_\rho$. Now, suppose $\gamma$
  can be partitioned into a family of disjoint regions $\cal F$.
  Then $\sum_{r \in {\cal F}}\ras(r) = C^\gamma_\rho$.
\end{corollary}

\subsection{The degenerate case}
Clearly, so far we derive the closed form of path integration based on
the assumption that$s_t^p \neq s_c^p$. The case where $s_t^p = s_c^p = s^p$
(note that we introduced a new notation $s^p$) is the degenerate case.
In that case, only mass changes from $w_c^p$ to $w_t^p$ but the volume
does not change. Therefore the contribution of numerator to the change
degenerates to (precisely it is only that the path integration changes):
\begin{align}
  \label{eq:degenerate-numerator-attribution}
  (w_t^r-w_c^r)\int_0^1\frac{1}{s_c^p}dt = \frac{w_t^r - w_c^r}{s^p}
\end{align}
For denominator, the attribution degenerates to
\begin{align}
  \label{eq:degenerate-denominator-attribution}
  -\frac{(s_t^r-s_c^r)}{(s_c^p)^2}\int_0^1\Big((w_t^p-w_c^p)t+w_c^p\Big)dt
  = -\frac{(s_t^r-s_c^r)(w_t^p+w_c^p)}{2(s_c^p)^2}
  = -\frac{(s_t^r-s_c^r)(w_t^p+w_c^p)}{2(s^p)^2}
\end{align}
Summing up this gives:
\begin{align}
  \label{eq:degnerate-region-attribution}
  \frac{w_t^r - w_c^r}{s^p} - \frac{(s_t^r-s_c^r)(w_t^p+w_c^p)}{2(s^p)^2}
\end{align}

\subsection{Teasing out churn effects}

In this section we consider the following setup, for every region we can
identify a set of ``entities'' where we can categorize traffic in a region
by these entities. Specifically, we consider three kinds of traffic:
\begin{enumerate}
\item Entities that only appear in the control, but not in test
  (control-only, shorthand CO below).
\item Entities that only appear in the test, but not in control
  (test-only, shorthand TO below).
\item Entities that appear in both test and control
  (test-control, shorthand TC below).
\end{enumerate}
For certain applications, it is important to do attribution separately for
each type of the traffic.

There is a simple reduction to reduce this attribution to our derivation
above. Basically, for every region $r$ (think about
$r=$(Domain=`facebook.com') above, this region is (implicitly) split
into three finer grained regions $r^{\rm CO}$, $r^{\rm CT}$, and
$r^{\rm TO}$. These three regions have disjoint traffic (because it is
identified by entities). Therefore, we can apply the derivation above
to each of the traffic in these three regions to compute the attributions,
respectively.

%% file: fim-model-code.tex
We describe FrequentItemsetModel, which can be used in cube crawling to
solve the classic Frequent Itemset Mining problem.
To this end, recall that the input data to frequent itemset mining is
a set of transactions, encoded as follows:

\begin{table}[htb]
  \centering
  \subfloat[]{
  \centering\small
    \begin{tabular}{|l|l|}
      \toprule
      \textbf{Id} & \textbf{Transaction}\\\hline\hline
      1 & \{ Beer, Diaper\} \\ \hline
      2 & \{ Beer, Diaper \} \\ \hline
      3 & \{ Rum, Whisky\} \\ \hline
      ... & ... \\\hline
      \bottomrule
    \end{tabular}
    \label{tab:transaction-dataset}
  }
  \subfloat[]{
  \centering\small
  \begin{tabular}{|l|l|l|l|l|}
      \toprule
      \textbf{Beer} & \textbf{Diaper} & \textbf{Rum} & \textbf{Whisky} & \textbf{Count}\\\hline\hline
       $\pmb{1}$ & $\pmb{1}$ & 0 & 0 & $\pmb{1}$ \\ \hline
       $\pmb{1}$ & $\pmb{1}$ & 0 & 0 & $\pmb{1}$ \\ \hline
       0 & 0 & $\pmb{1}$ & $\pmb{1}$ & $\pmb{1}$ \\ \hline
      ... & ... & ... & ... & ... \\\hline
      \bottomrule
    \end{tabular}
    \label{tab:encoding-df}
  }
  \caption{\small
  (a): An example transaction data set,
  for demonstrating FrequentItemsetModel.
  (b): `encoding\_df' that encodes a
  transaction data set, so that the ``items''
  in the transaction data set become dimensions
  in the encoding\_df. The first three rows
  are encoding of the transactions of
  the first three transactions in 
  Table~\ref{tab:transaction-dataset}.
}
  \label{tab:transactions-and-encoding}
\end{table}

To fit the problem into cube crawling, we first encode the items
as dimensions of a new table, `encoding\_df'.
Such encoding can be achieved using common libraries 
(e.g., TransactionEncoder in mlextend~\cite{raschkas_2018_mlxtend}).
Note that we attach a measure column, `Count', so that each row
(i.e. a transaction) has a count 1 (this can be generalized
by attaching arbitrary weight, instead of a count).
We can therefore now define a BaseTableAggregatableCube $\bfC_{\text{FIM}}$,
where: \textbf{(i)} The base table is `encoding\_df', and
\textbf{(ii)} The aggregation for Count is SUM.

\textbf{FrequentItemsetModel}.
Recall that for an itemset (a subset of all items), frequent itemset mining
computes a metric, called $\sf support$, as the fraction of transactions that contain this itemset.
To compute this with cube crawling, we observe that an itemset corresponds exactly to a region in the `encoding\_df',
where the dimensions in the region are exactly items in the itemset,
and all dimensions in the region takes value $\pmb{1}$.
We call such regions the ``itemset regions''.
Further, for an itemset region, such as (Beer=$\pmb{1}$, Diaper=$\pmb{1}$),
the aggregated Count (i.e., SUM(Count)) gives exactly the number of times
itemset \{Beer, Diaper\} appears in all transactions.
This thus gives FrequentItemsetModel:
\textbf{(1)} The model requests Count as a measure feature.
Note that, for example, $\bfC_{\text{FIM}}({\rm Count})$
is precisely the total number of transactions $N$ (i.e., number of rows) in the base table `encoding\_df'.
\textbf{(2)} At a region $\gamma$, If the region $\gamma$ contains a dimension
that is assigned value $0$, then we simply stop exploration from there,
since that region does not correspond to a itemset.
\textbf{(3)} Otherwise, the region $\gamma$ corresponds to an itemset,
and the region feature Count is the number $N_\gamma$ of number of
transactions that contain this itemset.
\textbf{(4)} Finally, $\sf support$ of this itemset is simply $N_\gamma/N$.
As a result, CubeCrawl(FrequentItemsetModel) computes the support of all itemset regions, as desired.
Detailed code is given below.

\begin{lstlisting}[language=Python]
class FrequentItemsetMiningModel(region_analysis_model.RegionAnalysisModel):
  # Initializer.
  def __init__(self, count_column, signal_names, apriori_signal_names):
    super().__init__(
        signal_names=signal_names,
        apriori_signal_names=apriori_signal_names,
        prefix=prefix)

    self._count_column = count_column
    self._metric_feature_specifiers = [count_column]
    self._total_count = 0

  @property
  def population_df(self):
    """Returns population data with which we compare with data."""
    return self._population_df

  @population_df.setter
  def population_df(self, population_df):
    self._population_df = population_df
    self._total_count = self._population_df[self._count_column]

  def __call__(self, data_frame, region):
    """Returns an Signals object for FIM.

    Args:
      data_frame: The data frame to compute attention signals. It should be
        noted that a wide table must be used as data_frame for this model.
      region: A tuple of attribute-value pairs representing the region to
        evaluate, e.g.: (('beer', True), ('Rum', False)).
    """
    signals = signals.Signals(support=0, continue_exploration=1)
    signals.signal_names = self._signal_names

    if any(not value for _, value in region):
      signals['continue_exploration'] = 0
    else:
      signals['continue_exploration'] = 1
      if region:
        signals['support'] = (
            data_frame[self._count_column].sum() / self._total_count)
      else:
        # Set the support of order-0 region to 0.
        signals['support'] = 0
    return signals
\end{lstlisting}

%% file: olapcomparison.tex
In this section, we will compare the standard OLAP framework with the HoCA framework introduced in this paper. First, we present the OLAP framework and its foundational building blocks.

\textbf{Online analytical processing (OLAP, for short)}.
OLAP broadly refers to the technology that enables fast \emph{multi-dimensional analytics} (MDA).
In MDA, the data schema is divided into two parts: categorical attributes (referred to as the dimensions) and numeric attributes (called measures).
A typical OLAP query seeks to understand the relationship of the measures with respect to a set of dimensions.
For example, consider a business warehouse with dimensions as \textsf{part}, \textsf{supplier}, \textsf{customer}, \textsf{date}.
The measure of interest is the total sales \textsf{sales}. An OLAP query over this warehouse may consist of aggregating the \textsf{sales}
for each (\textsf{part}, \textsf{supplier}, \textsf{customer}) combination.
A business user may look at the sales data and decide to breakdown the sales figure by \emph{rolling-up} to only the \textsf{customer} dimension
and look at the total sales figure per customer. Similarly, it is also possible to add the \textsf{date} dimension (this is known as \emph{drilling-down})
to understand the sales figures for each customer per month/quarter/year.

The fundamental abstraction that allows for fast computation of OLAP queries is efficient processing of \emph{data cubes}.
Given a set of dimensions $D$ and a measure $M$, a user may ask an OLAP query that corresponds to a \texttt{GROUP BY} query with a subset of $D$
as the group by columns and an aggregation function over $M$. The number of all possible group-by queries is $2^{|D|}$.
Further, the cardinality of the result of each group-by query is also exponential in the size of the input table (in the worst case).
In typical applications, the materialized result can be of the order of several hundred gigabytes,
so development of efficient data cube implementation algorithms is extremely critical.
The problem of selecting data cubes to materialize in order to obtain fast query response time
while minimizing the cube computation resource requirement has been the subject of intense research in the database community.

\smallskip
\textbf{HoCA vs. OLAP}. We now highlight the key differences between our framework and OLAP. 

\textbf{\em Analytical problems}.
OLAP is designed for fast multidimensional analytics which are mainly on the ``forward'' direction,
in the sense that users instantiates actions on \emph{where to examine next}. By contrast,
HoCA is designed to solve \emph{search} problems over the cube view of data formed by the region expressed in the program.
The analytical problems solved by HoCA are usually in the ``backward'' direction:
For example in the CostPerClick attribution case, HoCA starts with the \emph{top-level} CostPerClick change,
and the analysis deconstructs the metric and search for the regions
that are major contributors to this change.
To solve the same Cpc attribution problem in OLAP, users have to use OLAP operators to perform an iterative,
but \emph{manual}, slice-and-dice analysis to find the insights.

\textbf{\em HoCA cubes vs. OLAP cubes}.
This difference in analytic questions induces vastly different considerations of
the cube design: In HoCA, cubes are designed to be a \emph{logical form} of data,
and HoCA operators are about higher-order, cube-to-cube transformations. These operators
can be composed in a HoCA program, where HoCA cubes flow in between. 
OLAP cubes are, on the other hand, designed for entire languages (e.g.~\cite{MDX})
for forward cube analytic capabilities. OLAP cubes are monolithic data objects,
and there is little sense of OLAP cubes explicitly flowing between operators.
Since one of the goals of HoCA is to build a framework that is easy to program in,
abstract cube and regions provides an intuitive way to model the sub-space that needs to be explored.
Further, observe that the notion of cellset is general enough to capture an OLAP data cube.
Indeed, the * in a Cellset can be seen as an equivalent of \texttt{ANY} in OLAP cubes\footnote{
Note that the cardinality of a data cube (which is $\Pi_{i \in [n]} (U_i+1)$) also matches that of a Cellset}.
A cellset is also more expressive than an OLAP cube since we do not require the measure to be generated
by a \emph{decomposable} aggregate function, which is typically the case in OLAP.
This property has important implications for compositions in HoCA,
in that given a materialized data cube,
one could interpret that as a CellsetCube, allowing users to write programs over the cube subsequently.

\textbf{\em Analysis and programming}.
HoCA provides a data model (a new logical form of data) and an algebraic framework.
HoCA provides novel \emph{programming experiences} and \emph{constructions},
such as Region Analysis Model programming, cube join, and recurrent cube crawling,
which are all beyond OLAP considerations.
To this end, we emphasize that HoCA cube crawling allows passing of information (via the population data for each model)
about the entire region induced by all the dimensions to the sub-regions crawled.
This is important to allow for computation of advanced signals such as {\tt aumann\_shapley\_attribution}
(which requires such population information, see Section~\ref{sec:density-model-derivation}). 
OLAP does not have a counterpart to allow for message passing between sub-cubes.

\textbf{A remark on similarities between HoCA and OLAP}.
We note a key similarity between the two frameworks that are likely
to allow several optimizations that have been designed for OLAP to improve HoCA's efficiency.
For a BaseTableGroupByCube, given a set of dimensions $D' = \{d_1, \dots d_k\} \subseteq D$, the number of the region spaces explored by HoCA
and the cardinality of the data cube formed by $D'$ is the same. This suggests that ROLAP based optimizations built for exploiting the overlap
between different data cubes could potentially be transferred to multiple HoCA programs exploring common regions.
Indeed, multiple HoCA programs running over the same input can be optimized jointly to minimize computation.
We leave the problem of building such an optimizer as a problem for future work.

\subsection{Expressing Constraint Verification using cube crawling in HoCA}

In this section, we will show how the problem of detecting integrity constraint violations can be naturally expressed as a HoCA program.
We will consider the following version of functional dependency detection: Consider a set of columns $\mathcal{S}$, and another set of
columns $\mathcal{T}$ (for concreteness, let $\mathcal{S} = \{X, Y\}$, and $\mathcal{T} = \{Z, W\}$). Now given a relation $R$,
we want to check whether there is a functional dependency $\mathcal{S} \rightarrow \mathcal{T}$ in $R$, that is:

\begin{align*}
  \Big[\forall t, t' \in R\Big] t[\mathcal{S}] = t'[\mathcal{S}] \Rightarrow t[\mathcal{T}] = t'[\mathcal{T}]
\end{align*}

That is, for any two tuples $t, t' \in R$ that have the same values when projected over attributes $\mathcal{S}$,
the projections of the two tuples on attribute $\mathcal{T}$ are also the same.
We show that there are various ways to program this using cube crawling in HoCA
(we will use some pseudo-code to simplify things and highlight ideas)

\textbf{Approach 1: Encoding counting as a cube metric}.
In this approach, we push the counting work into the cube abstraction, using a BaseTableGroupByCube.
For this, consider we define the following metric:

\begin{lstlisting}[language=Python]

distinct_target_count = COUNT(DISTINCT Z, W)
\end{lstlisting}

where $Z, W$ are the assumed members in $\mathcal{T}$. Therefore, for any region say $\gamma: (X=x, Y=y)$,
if we request the metric feature {\tt distinct\_target\_count}, we get the number of projections of $Z, W$ conditioned
on the values in $\mathcal{S}$. With this, we can consider a trivial model {\tt IdModel}, which simply requests
the metric feature {\tt distinct\_target\_count}, as follows:

\begin{lstlisting}[language=Python]

id_model = IdModel(metric_features=['distinct_target_count'])
\end{lstlisting}

Therefore, if {\tt distinct\_target\_count} is larger than 1 for any region, there is no functional dependency in $R$.
This can be encoded using the following cube crawling:

\begin{lstlisting}[language=Python]

CubeCrawl(
  data_cube=BaseTableGroupByCube(R),
  grouping_sets=[['X', 'Y'],],
  models=[id_model],
  thresholds = {distinct_target_count : 1.0})
\end{lstlisting}

Note that here we used the {\tt grouping\_sets} parameter to focus on grouping set $[X, Y]$, because we don't need
to consider any region of smaller degree.

\textbf{Approach 2: Counting entities in a model}.
In the approach, we program an region analysis model called the {\tt entity\_model}.
This model receives as parameters a set of dimensions called {\tt entity\_columns},
which in this case is just $\mathcal{T}$. Upon a region, the BaseTableGroupByCube projects the data to the region
$\gamma: (X=x, Y=y)$, and requests two dimension columns $\mathcal{T} = [Z, W]$. The semantics is simply computing
the group-by on these dimensions and return a table of the $Z, W$ pairs. Because of the groupby, these two columns
become the primary keys. Therefore, by counting the number of entities in the region feature table, say recording it
in a signal called {\tt entity\_count}, we can detect functional dependency failure if any region has
{\tt entity\_count} larger than one. This gives the following HoCA program:

\begin{lstlisting}[language=Python]
entity_model = EntityModel(
  entity_columns=['Z', 'W'],
  output_signals=['entity_count'])

CubeCrawl(
  data_cube=BaseTableGroupByCube(R),
  grouping_sets=[['X', 'Y'],],
  models=[entity_model],
  thresholds = {entity_count : 1.0})\end{lstlisting}
  
\textbf{Approach 3: An alternative model using more standard SQL semantics}.
In this approach, we augment the base relation with a new column say {\tt const\_one},
where it is a constant 1, and we put a SUM aggregation function. Then the we can program a different
region analysis model, say {\tt EntityMeasureModel}, which requests a set of entity columns,
and a set of entity measures, and then process them. In our case, for each region,
{\tt entity\_measure\_model} will receives a table where each row contains an entity and its corresponding
measures. Therefore again by just counting the number entities (ignoring the aggregated measure),
and recording the answer using a signal called {\tt entity\_count},
we can detect functional dependency failure if any region has more than one entities.
This gives the following program:

\begin{lstlisting}[language=Python]
entity_model = EntityMeasureModel(
  entity_columns=['Z', 'W'],
  entity_measure=['const_one']
  output_signals=['entity_count'])

CubeCrawl(
  data_cube=BaseTableGroupByCube(R),
  grouping_sets=[['X', 'Y'],],
  models=[entity_model],
  thresholds = {entity_count : 1.0})\end{lstlisting}

%% file: more-examples.tex
\subsection{Using cube crawling to generate pseudo-labels for deep learning}
We describe a case where cube crawling plays a key role in deep learning.
Suppose that we have a set of ``events'' that we want to perform a binary classification
(so class labels are $0/1$). However, unlike typical supervised learning setting,
there is no label on the training data (or in other words, the training data is entirely unlabeled).
In this case, one general strategy is the following:
\textbf{(1)} We first generate pseudo-labels (we call them plabels in the following),
\textbf{(2)} We then use these plabels together with the unlabeled training data to
train a deep learning model.
Now, the question reduces to: \emph{How to generate good plabels?}

\begin{figure}[htb]
  \centering
  \includegraphics[width=0.5\textwidth]{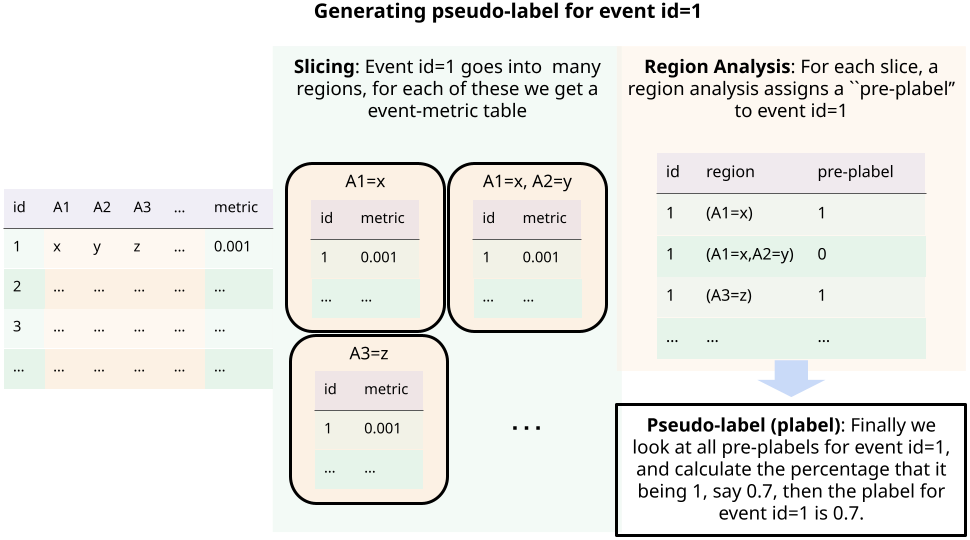}
  \caption{\small
  Cube crawling used for generating pseudo-labels (plabels)
  for a downstream deep-learning pipeline.}
  \label{fig:crawling-pseudolabels-dl}
\end{figure}

To answer the question, we note that the training data we have is organized as a relational table,
as shown on the left side of Figure~\ref{fig:crawling-pseudolabels-dl}.
Each row represents an event, and an event is described by its id, a bunch of attributes
(abstracted into `A1', `A2', ...,) and a single metric `metric' (to simplify things).
Now the idea is to use cube crawling to generate plabels: By considering different regions of the
abstract cube (induced by slicing on the base table on attributes), we can put a single event
(say event id=1) into different regions (``Slicing'' in Figure~\ref{fig:crawling-pseudolabels-dl}).
For each region we run a region analysis to check the behavior of all events in that region,
and based on that we assign `pre-plabels' (either 0 or 1) to all events in the region
(``Region Analysis'' in Figure~\ref{fig:crawling-pseudolabels-dl}).
After this step, for each event, we have an ``event table'' of its pre-plabels in different regions.
Finally, the plabel of an event is just the fraction that its pre-plabel is 1 in the respective event table.
Note that this is ``soft label'' which somehow captures the ``probability'' that we thinks it is 1.
In practice, we found that considering a large number of regions is crucial for this strategy to work,
since it gives different ``view points'' about the behavior of an event. Therefore the efficiency of
crawling is important. Also, our design allows machine learnists to focus on writing the statistical
analysis to assign pre-plabels.

\subsection{Cross rank correlation}
\input{cross-rank-correlation/cross-rank-correlation}

%% file: cross-rank-correlation/cross-rank-correlation.tex
For the test period, we have $N$ data points of (Budget, Cost) pairs
$U = \{(B^t_i, R^t_i)\}_{i=1}^N$. Similarly, for the control period, we have
$V = \{B^c_j, R^c_j\}_{j=1}^M$ (note that there are $M$ points, and $N$ may not be equal to $M$).
The correlation formula is the following:

\begin{align}
  \label{eq:cross-rank-correlation}
  \text{CrossRankCorr}(U, V) = \frac{\sum_{i \in [N], j \in [M]}\text{sgn}\left(\Big(R^t_i-R^c_j\Big)\Big(B^t_i - B^c_j\Big)\right)}{NM}
\end{align}